\documentclass[12pt,draftclsnofoot,onecolumn]{IEEEtran}
\usepackage[latin9]{inputenc}
\usepackage{amsmath}
\usepackage{amssymb}
\usepackage{graphicx}
\usepackage{esint}

\makeatletter

\IEEEoverridecommandlockouts
\usepackage{cite}

\newtheorem{proposition}{Proposition}
\newtheorem{corollary}{Corollary}

\newtheorem{lemma}{Lemma}

\allowdisplaybreaks[4]

\makeatother

\begin{document}

\title{Secure Communication Via an Untrusted Non-Regenerative Relay in
  Fading Channels 
\thanks{The work was supported by the U.S. Army
    Research Office MURI grant W911NF-07-1-0318, and by the National
    Science Foundation under grant CCF-1117983.}  }

\author{
\IEEEauthorblockN{Jing~Huang, Amitav~Mukherjee, and
  A.~Lee~Swindlehurst} \\
\IEEEauthorblockA{\normalsize Electrical Engineering and Computer Science \\
University of California, Irvine, CA 92697 \\
Email: \{jing.huang; amukherj; swindle\}@uci.edu }
}

\maketitle
%\begin{titlepage}

\begin{abstract}
  We investigate a relay network where the source can potentially utilize an untrusted non-regenerative relay to augment its direct transmission of a confidential message to the destination. Since the relay is untrusted, it is desirable to protect the confidential data from it while simultaneously making use of it to increase the reliability of the transmission.  We first examine the secrecy outage probability (SOP) of the network assuming a single antenna relay, and calculate the exact SOP for three different schemes: direct transmission without using the relay, conventional non-regenerative relaying, and cooperative jamming by the destination.  Subsequently, we conduct an asymptotic analysis of the SOPs to determine the optimal policies in different operating regimes. We then generalize to the multi-antenna relay case and investigate the impact of the number of relay antennas on the secrecy performance. Finally, we study a scenario where the relay has only a single RF chain which necessitates an antenna selection scheme, and we show that unlike the case where all antennas are used, under certain conditions the cooperative jamming scheme with antenna selection provides a diversity advantage for the receiver. Numerical results are presented to verify the theoretical predictions of the preferred transmission policies.
\end{abstract}

%\end{titlepage}

\section{Introduction}

The broadcast characteristic of the wireless medium facilitates a
number of advanced communication protocols at the physical layer,
such as cooperative communications with the aid of relay nodes. Since
relays can overhear signals emanating from a source and rebroadcast
them towards the intended destination, the reliability of the transmission
can be improved via diversity. However, the broadcast property also
makes it difficult to shield information from being leaked to unintended
receivers (eavesdroppers), which has led to an intensive study of
improving information security at the physical layer of wireless networks
\cite{Liang_Information08}.

The foremost metric of physical-layer information security is the
secrecy capacity, which quantifies the maximal transmission rate at
which the eavesdropper is unable to decode any of the confidential
data. An alternative secrecy criterion that has recently been investigated
for fading channels is the secrecy outage probability (SOP), from
which one can determine the likelihood of achieving a certain
secrecy rate \cite{Bloch_Wireless08}.

In the context of single-input single-output (SISO) relay channels,
the SOP has been investigated in \cite{Krikidis_Relay09,Krikidis_Opportunistic10,Ding_Opportunistic11}
for networks composed of external eavesdroppers that are distinct
from the source/sink and relay nodes.  Secrecy may still be an issue
even in the absence of external eavesdroppers, since one may desire
to keep the source signal confidential from the relay itself in spite
of its assistance in forwarding
the data to the destination \cite{Oohama_Coding01}. The relay is
in effect also an eavesdropper, even though it complies with the source's
request to forward messages to the destination. For example, an \emph{untrusted
relay} may belong to a heterogeneous network without the same security
clearance as the source and destination nodes. This scenario has been
studied in \cite{He_Two-Hop09,He_Cooperation10} where the authors
presented bounds on the achievable secrecy rate. Furthermore, they
showed that non-regenerative or amplify-and-forward (AF) and compress-and-forward
relaying (including a direct link) admit a non-zero secrecy rate even
when the relay is untrusted, which does not hold for decode-and-forward
relaying.

Our paper analyzes a three-node relay network where the source
can potentially utilize a multi-antenna untrusted relay to augment
the direct link to its destination. In \cite{Jeong_Joint12}, the
authors considered the joint source/relay beamforming design problem
for secrecy rate maximization via an AF multi-antenna relay. In realistic
fading channels, the secrecy outage probability is a more meaningful
metric compared to the ergodic secrecy rate, which is ill-defined under
finite delay constraints. Thus, unlike \cite{Jeong_Joint12}, in this
work we focus on the SOP of the AF relaying protocol, which is chosen
due to its increased security vis-à-vis decode-and-forward and lower
complexity compared to compress-and-forward.

When multiple antennas are employed in relay networks, any potential
performance benefits must be balanced against increased hardware complexity
and power consumption. As a reduced-complexity solution that can maintain
full diversity, antenna selection has received extensive attention
in AF relay networks, for example in cases where only one RF chain
is available at the relay \cite{Amarasuriya_Feedback10}. In \cite{Zhang_Near-Optimal10},
a low-complexity near-optimal antenna selection algorithm was proposed
for maximizing the achievable rate. The bit error rate performance
obtained by choosing the best antenna pairs over both relay hops was
examined in \cite{Kim_End06}. However, the open problem addressed
in our paper is the tradeoff between the diversity gain for the legitimate
receiver versus the inadvertent diversity gain of the information
leaked to the untrusted relay in the first hop.

In this paper, we calculate the exact SOP with a multi-antenna
relay for three different transmission policies: (1) direct
transmission (DT) where the relay is considered as a pure
eavesdropper, (2) conventional AF relaying, and (3) cooperative
jamming (CJ) by the destination in the first hop to selectively
degrade the relay's eavesdropping capability.  There are
scenarios (with different SNR, number of antennas, channel gains,
\textit{etc.}) where each of these three schemes owns a
performance advantage over the other two. Our analysis allows
these performance transitions to be determined. We also conduct
an asymptotic analysis for the special case of a single-antenna
relay and elicit the optimal policies for different power budgets
and channel gains.  Secrecy is typically compromised with a
multi-antenna relay employed, especially if the number of relay antennas
grow large.  We show that the secrecy performance improves when
the relay can only perform antenna selection instead of
beamforming, but in nearly all cases the SOP still grows with the
number of relay antennas. A non-increasing SOP is shown to only
be obtained when the CJ scheme is used and the second-hop CSI can
be hidden from the relay.

The remainder of this work is organized as follows. The mathematical
models of the AF relaying and cooperative jamming approaches are introduced
in Section~\ref{sec:pf}. The secrecy outage probabilities of direct
transmission, AF relaying, and cooperative jamming with a multi-antenna
relay are examined in Section \ref{sec:ut}, and the specialization
to a relay employing antenna selection is presented in Section~\ref{sec:secrecy-with-antenna}.
Selected numerical results are shown in Section~\ref{sec:nr}, and
we conclude in Section~\ref{sec:con}.

\section{Mathematical Model}

\label{sec:pf}

We consider a half-duplex two-hop relaying network composed of a source
(Alice), a destination (Bob), and an untrusted relay that when active
employs the AF protocol. Alice and Bob are both single-antenna nodes,
and the relay is assumed to be equipped with $K$ antennas. The channel
is assumed to be quasi-static (constant during the two hops) with
Rayleigh fading and a direct link between Alice and Bob is assumed
to be available. We also assume all nodes in the network have the
same power budget $P$.

\subsection{Relay Protocol}

For AF relaying, during the first phase the relay receives
\begin{eqnarray}
\mathbf{y}_{R}=\mathbf{h}_{A,R}x_{A}+\mathbf{n}_{R}\label{eq:1}
\end{eqnarray}
where $x_{A}$ is the zero-mean signal transmitted by Alice with variance
$\mathbb{E}\{x_{A}^{H}x_{A}\}=P$, $\mathbf{h}_{A,R}=[h_{A,1},h_{A,2},\dots,h_{A,K}]^{T}$
is the complex circularly symmetric Gaussian channel vector with covariance
matrix $\mathbf{C}_{A,R}=diag\{\bar{\gamma}_{A,1},\bar{\gamma}_{A,2},\dots,\bar{\gamma}_{A,K}\}$,
and $\mathbf{n}_{R}$ is additive white Gaussian noise with
covariance $N_{0}\mathbf{I}$.  For the moment, we assume spatially
white noise with no jamming present; the case of cooperative jamming will be discussed separately.  In general, we will use $\mathbf{h}_{i,j}$
to represent the channel vector between node $i$ and $j$, with $i,j\in\{A,B,R\}$
denoting which of the terminals is involved. Let $\gamma_{i,j}\triangleq|h_{i,j}|^{2}$
be the instantaneous squared channel strength, so that $\gamma_{i,j}$
is exponentially distributed with hazard rate $\frac{1}{\bar{\gamma}_{i,j}}$,
\textit{i.e.} $\gamma_{i,j}\sim\exp\left(\frac{1}{\bar{\gamma}_{i,j}}\right)$.
The probability density function (p.d.f.) of $\gamma_{i,j}$ is given
by
\begin{equation}
p_{\gamma_{i,j}}(x)=\frac{1}{\bar{\gamma}_{i,j}}~e^{-\frac{x}{\bar{\gamma}_{i,j}}},\quad x\ge0.\label{eq:2}
\end{equation}

When the relay has $K$ RF chains and can implement beamforming, we
assume that, for purposes of forwarding the message, the relay adopts maximum ratio combining (MRC) on receive and maximum ratio transmission (MRT) on transmit.  Thus, the output of the relay receiver is given by
\[
\tilde{y}_{R}=\sqrt{\sum\nolimits _{i=1}^{K}|h_{A,i}|^{2}}x_{A}+\frac{\sum_{i=1}^{K}h_{A,i}^{H}n_{i}}{\sqrt{\sum_{i=1}^{K}|h_{A,i}|^{2}}}.
\]
The relay then transmits the signal $\mathbf{x}_{R}=\frac{\sqrt{P}}{\sigma}\frac{\mathbf{h}_{R,B}^{H}}{||\mathbf{h}_{R,B}||}\tilde{y}_{R}$,
where $\sigma=\sqrt{\mathbb{E}\{||\tilde{y}_{R}||^{2}\}}$ and $\mathbf{h}_{R,B}=[h_{1,B},h_{2,B},\dots,h_{K,B}]$.
The received signal at Bob over both phases is then given by
\begin{equation}
\mathbf{y}_{B}=\left[\begin{array}{c}
h_{A,B}\\
\frac{\sqrt{P}}{\sigma}\sqrt{\sum_{i=1}^{K}|h_{i,B}|^{2}}\sqrt{\sum_{i=1}^{K}|h_{A,i}|^{2}}
\end{array}\right]x_{A}+\left[\begin{array}{c}
n_{B1}\\
\frac{\sqrt{P}}{\sigma}\sqrt{\frac{\sum_{i=1}^{K}|h_{i,B}|^{2}}{\sum_{i=1}^{K}|h_{A,i}|^{2}}}\sum_{i=1}^{K}h_{A,i}^{H}n_{i}+n_{B2}
\end{array}\right],\label{eq:4}
\end{equation}
where we assume that $n_{B1}$ and $n_{B2}$ are uncorrelated Gaussian
noise with variance $N_{0}$. The subscripts 1 and 2 refer to the
first and second transmission phases, respectively. Since the antennas
on the relay are much closer together compared to their distances
to the source and the destination, we assume $\{\bar{\gamma}_{A,i}\}_{i=1}^{K}=\bar{\gamma}_{A,R}$
and $\{\bar{\gamma}_{i,B}\}_{i=1}^{K}=\bar{\gamma}_{R,B}$.

We assume Alice uses a codebook
$\mathcal{C}(2^{nR_{0}},2^{nR_{s}},n)$ where $R_{s}$ is the
intended secrecy rate ($R_{s}\le R_{0}$), $n$ is the codeword
length, $2^{nR_{0}}$ is the size of the codebook, and
$2^{nR_{s}}$ is the the number of confidential messages to
transmit.  The $2^{nR_{0}}$ codewords are randomly grouped into
$2^{nR_{s}}$ bins. To send confidential message
$w\in\{1,\cdots,2^{nR_{s}}\}$, Alice will use a stochastic
encoder to randomly select a codeword from bin $w$ and send it
over the channel. Since in our model the untrusted relay only
wiretaps in the first phase, and since the AF and CJ schemes are
mathematically equivalent to a one-stage $1\times2$ SIMO wiretap
channel \cite{Dong_Improving10,He_Cooperation10}, a conventional
wiretap code can be applied.  The achievable secrecy rate is then
the same as for the equivalent one-hop channel: \emph{i.e.,}
$R_{s}=[I_{B}-I_{E}]^{+}$ where $[x]^{+}\triangleq\max\{0,x\}$,
$I_{B}$ is the mutual information for the legitimate link and
$I_{E}$ is the mutual information for the eavesdropper link.

\subsection{Cooperative Jamming}

\label{sec:cooperative-jamming}

Various cooperative jamming schemes involving the transmission of
artificial interference have been proposed in previous work for improving
secrecy \cite{Ding_Opportunistic11,Tekin_General08,Goel_Guaranteeing08,Huang_Cooperative11,Huang_Robust11}.
In this paper, as an alternative to the traditional AF protocol, we
assume a half-duplex cooperative jamming scheme where Bob forfeits
information from Alice during the first phase in favor of transmitting
a jamming signal. Under this model, the received signal at the relay
is
\begin{equation}
\mathbf{y}_{R}=\mathbf{h}_{A,R}x_{A}+\mathbf{h}_{R,B}z_{B}+\mathbf{n}_{R}\label{eq:5}
\end{equation}
where $z_{B}$ is the jamming signal transmitted by Bob with power
$\mathbb{E}\{z_{B}^{H}z_{B}\}=P$. Similar to the AF scheme, during
the second phase, the relay scales $\mathbf{y}_{R}$ and forwards
it to Bob, and thus the received signal at Bob can be written as
\begin{equation}
y_{B}=\frac{\sqrt{P}}{\sigma}\sqrt{\sum_{i=1}^{K}|h_{i,B}|^{2}\sum_{i=1}^{K}|h_{A,i}|^{2}}\, x_{A}+\frac{\sqrt{P}}{\sigma}\sqrt{\frac{\sum_{i=1}^{K}|h_{i,B}|^{2}}{\sum_{i=1}^{K}|h_{A,i}|^{2}}}\left(\sum_{i=1}^{K}h_{i,B}^{H}h_{A,i}z_{B}+\sum_{i=1}^{K}h_{A,i}^{H}n_{i}\right)+n_{B2}\label{eq:6}
\end{equation}
where we assume a reciprocal channel between the relay and Bob: $\mathbf{h}_{R,B}=\mathbf{h}_{B,R}^{H}$.
Note that the intentional interference term can be removed by Bob
since $z_{B}$ is known to him.

While we assume that the relay forwards the output of the MRC beamformer to Bob, when computing the mutual information available to the relay
we assume that she can perform MMSE (maximum SINR) beamforming to
counteract the jamming from Bob.  Recall that since Bob is the source of the interference, he can subtract its contribution from the signal forwarded by the relay regardless of the choice of receive beamformer employed at the relay.

\section{Transmission with an Untrusted Relay}

\label{sec:ut}

\subsection{Single-antenna Relay}

\label{sec:single-antenna-relay}

We begin with the case where the relay employs a single antenna.
In this section, we will calculate the exact SOP expressions for the
DT, AF and CJ schemes and analyze the corresponding asymptotic behavior
under limiting conditions on the power budgets and channel gains.

\subsubsection{DT}

Direct transmission refers to the case where Alice uses a single-hop
transmission to communicate with Bob rather than cooperating with
the relay. As illustrated later, in some cases this strategy provides
better secrecy performance than AF and CJ. Under DT, the relay is simply treated
as a pure eavesdropper. Thus the model will be simplified to a traditional
wiretap channel with Rayleigh fading, which has been fully characterized
in \cite{Bloch_Wireless08}, for example. Since the channel gains
are assumed to be quasi-static, the achievable secrecy rate for one
channel realization is given by
\begin{equation}
R_{s}^{DT}=[I_{B}^{DT}-I_{R}^{DT}]^{+}\label{eq:7}
\end{equation}
where $I_{B}^{DT}$ and $I_{R}^{DT}$ represent the mutual information
between Alice and Bob, and between Alice and the relay respectively,
and are given by $I_{B}^{DT}=\log_{2}(1+\frac{P}{N_{0}}|h_{A,B}|^{2})$
and $I_{R}^{DT}=\log_{2}(1+\frac{P}{N_{0}}|h_{A,R}|^{2})$.

When $K=1$, the probability of a positive secrecy rate is given by
\cite{Bloch_Wireless08}
\begin{align}
\mathcal{P}_{pos}^{DT} & =\mathcal{P}(I_{B}^{DT}-I_{R}^{DT}>0)=\mathcal{P}(|h_{A,B}|^{2}>|h_{A,R}|^{2})\notag\nonumber \\*
 & =\frac{\bar{\gamma}_{A,B}}{\bar{\gamma}_{A,R}+\bar{\gamma}_{A,B}}.\label{eq:8}
\end{align}
It is interesting to note that, in the presence of fading, a nonzero
secrecy rate exists even when $\bar{\gamma}_{A,R}>\bar{\gamma}_{A,B}$,
\textit{i.e.} when the eavesdropper's channel is on average better
than the legitimate channel \cite{Bloch_Wireless08} (although the
probability of such an event is less than 1/2). Eq.~\eqref{eq:8}
also indicates that Alice will be unable to reliably transmit secret
messages when $\bar{\gamma}_{A,R}\rightarrow\infty$, \textit{e.g.},
when the untrusted relay is proximate to Alice.

The outage probability for a target secrecy rate $R$ is given by
\cite{Bloch_Wireless08}
\begin{align}
\mathcal{P}_{out}^{DT}(R) & =\mathcal{P}\{I_{B}^{DT}-I_{R}^{DT}<R\}=\mathcal{P}\left\{ \frac{1+\rho|h_{A,B}|^{2}}{1+\rho|h_{A,R}|^{2}}<2^{R}\right\} \notag\nonumber \\
 & =1-\frac{\bar{\gamma}_{A,B}}{2^{R}\bar{\gamma}_{A,R}+\bar{\gamma}_{A,B}}e^{-\frac{2^{R}-1}{\rho\bar{\gamma}_{A,B}}},\label{eq:9}
\end{align}
where $\rho\triangleq\frac{P}{N_{0}}$ is the transmit SNR. The secrecy
outage probability is a criterion that indicates the fraction of fading
realizations where a secrecy rate $R$ cannot be supported, and also
provides a security metric for the case where Alice and Bob have no CSI for the
eavesdropper \cite{Bloch_Wireless08}. However, we recognize the alternative
definition of SOP that was recently proposed in \cite{Zhou_SOP11},
which provides a more explicit measurement of security level by only
considering secrecy outage events conditioned on a reliable legitimate
link. The secrecy outage results presented in this work can be reformulated
according to this alternative definition in a straightforward manner.

\subsubsection{AF}

When the untrusted AF relay is employed for cooperation, the channel
is equivalent to the conventional wiretap channel where Bob receives
the signal from two orthogonal channels \cite{He_Cooperation10},
and thus the achievable secrecy rate can be computed from $R_{s}^{AF}=\left[I_{B}^{AF}-I_{R}^{AF}\right]^{+}$,
where
\begin{equation}
I_{B}^{AF}=\frac{1}{2}\log_{2}\left(1+\rho|h_{A,B}|^{2}+\rho\frac{|h_{R,B}|^{2}|h_{A,R}|^{2}}{|h_{R,B}|^{2}+\bar{\gamma}_{A,R}+\frac{1}{\rho}}\right)\label{eq:10}
\end{equation}
and
\begin{equation}
I_{R}^{AF}=\frac{1}{2}\log_{2}\left(1+\rho|h_{A,R}|^{2}\right).\label{eq:11}
\end{equation}
Therefore, the probability of achieving a positive secrecy rate for
AF relaying is formulated as
\begin{equation}
\mathcal{P}_{pos}^{AF}=\mathcal{P}\left\{ |h_{A,B}|^{2}+\frac{|h_{R,B}|^{2}|h_{A,R}|^{2}}{|h_{R,B}|^{2}+\bar{\gamma}_{A,R}+\frac{1}{\rho}}>|h_{A,R}|^{2}\right\} .\label{eq:12}
\end{equation}
In Appendix \ref{sec:prob-posit-secr}, we show that this probability
is given by
\begin{equation}
\mathcal{P}_{pos}^{AF}=\mu_{1}(\beta_{1}-1)e^{\mu_{1}\beta_{1}}\mathrm{Ei}(-\mu_{1}\beta_{1})+1\label{eq:13}
\end{equation}
where $\mu_{1}=\frac{\bar{\gamma}_{A,R}+1/\rho}{\bar{\gamma}_{R,B}}$,
$\beta_{1}=1+\frac{\bar{\gamma}_{A,R}}{\bar{\gamma}_{A,B}}$ and $\mathrm{Ei}(\cdot)$
is the exponential integral $\mathrm{Ei}(x)=\int_{-\infty}^{x}e^{t}t^{-1}~dt$.

The outage probability of the AF scheme for a given secrecy rate $R$
can be written as
\begin{equation}
\mathcal{P}_{out}^{AF}(R)=\mathcal{P}\left\{ \frac{1+\rho|h_{A,B}|^{2}+\rho\frac{|h_{R,B}|^{2}|h_{A,R}|^{2}}{|h_{R,B}|^{2}+\bar{\gamma}_{A,R}+\frac{1}{\rho}}}{1+\rho|h_{A,R}|^{2}}<2^{2R}\right\} ,\label{eq:14}
\end{equation}
and the exact SOP is given by the following proposition.

\begin{proposition} \label{sec:af} The secrecy outage probability
for AF relaying can be expressed as
\begin{equation}
\mathcal{P}_{out}^{AF}(R)=1-\frac{\bar{\gamma}_{A,B}}{(2^{2R}-1)\bar{\gamma}_{A,R}+\bar{\gamma}_{A,B}}e^{-\frac{2^{2R}-1}{\rho\bar{\gamma}_{A,B}}}\left[\mu_{1}(\beta_{2}-1)e^{\mu_{1}\beta_{2}}\mathrm{Ei}(-\mu_{1}\beta_{2})+1\right]\label{eq:15}
\end{equation}
where $\mu_{1}=\frac{\bar{\gamma}_{A,R}+1/\rho}{\bar{\gamma}_{R,B}}$,
$\beta_{2}=\frac{2^{2R}\bar{\gamma}_{A,R}+\bar{\gamma}_{A,B}}{(2^{2R}-1)\bar{\gamma}_{A,R}+\bar{\gamma}_{A,B}}$,
and $R$ is the target secrecy rate. \end{proposition}
\begin{IEEEproof}
See Appendix \ref{sec:secr-outage-prob}.
\end{IEEEproof}
For the high SNR regime, \eqref{eq:14} can be approximated as
\begin{equation}
\mathcal{P}_{out}^{AF}(R)\simeq\mathcal{P}\left(\frac{|h_{A,B}|^{2}+\frac{|h_{R,B}|^{2}|h_{A,R}|^{2}}{|h_{R,B}|^{2}+\bar{\gamma}_{A,R}}}{|h_{A,R}|^{2}}<2^{2R}\right),\label{eq:16}
\end{equation}
which is a function independent of $\rho$. This indicates that the
AF scheme does not approach zero SOP even as the transmit power is
increased. Intuitively, this is reasonable since any increase in the
transmit power will bolster the SNR at both the legitimate user and
the eavesdropper.  The asymptotic value of the SOP at high SNR will be characterized in Section~\ref{sec:asymptotic-behavior}.

\subsubsection{CJ}

As mentioned above, for the CJ approach we assume Bob ignores the
direct link and transmits a jamming signal during the first
phase. According to the signal model in Section \ref{sec:cooperative-jamming},
we have the following expression for the mutual information between
Alice and Bob in this case:
\begin{equation}
I_{B}^{CJ}=\frac{1}{2}\log_{2}\left(1+\rho\frac{|h_{R,B}|^{2}|h_{A,R}|^{2}}{|h_{R,B}|^{2}+\bar{\gamma}_{A,R}+\bar{\gamma}_{R,B}+\frac{1}{\rho}}\right)\label{eq:17}
\end{equation}
\begin{equation}
I_{R}^{CJ}=\frac{1}{2}\log_{2}\left(1+\frac{|h_{A,R}|^{2}}{|h_{R,B}|^{2}+\frac{1}{\rho}}\right),\label{eq:18}
\end{equation}
and the corresponding probability of a positive secrecy rate is given
by
\begin{equation}
\mathcal{P}_{pos}^{CJ}=\mathcal{P}\left\{ \rho\frac{|h_{R,B}|^{2}}{|h_{R,B}|^{2}+\bar{\gamma}_{A,R}+\bar{\gamma}_{R,B}+\frac{1}{\rho}}>\frac{1}{|h_{R,B}|^{2}+\frac{1}{\rho}}\right\} =e^{-\frac{1}{\bar{\gamma}_{R,B}}\sqrt{\frac{\bar{\gamma}_{A,R}+\bar{\gamma}_{R,B}+\frac{1}{\rho}}{\rho}}}.\label{eq:19}
\end{equation}
From \eqref{eq:19}, we see that $\mathcal{P}_{pos}^{CJ}$ is a monotonically
increasing function of $\bar{\gamma}_{R,B}$ when $\bar{\gamma}_{A,R}$
is fixed. In other words, CJ is not appropriate when the second hop
channel is weak. This is not surprising since, when CJ is employed,
the half-duplex constraint for Bob means that the information from
the direct link is ignored, and Bob relies heavily on the second hop
to obtain the information from Alice.

The outage probability in this case can be expressed as
\begin{equation}
\mathcal{P}_{out}^{CJ}(R)=\mathcal{P}\left(\frac{1+\rho\frac{|h_{R,B}|^{2}|h_{A,R}|^{2}}{|h_{R,B}|^{2}+\bar{\gamma}_{A,R}+\bar{\gamma}_{R,B}+\frac{1}{\rho}}}{1+\frac{|h_{A,R}|^{2}}{|h_{R,B}|^{2}+\frac{1}{\rho}}}<2^{2R}\right),\label{eq:20}
\end{equation}
and the exact SOP expression is provided in the following proposition.

\begin{proposition} \label{sec:coop-jamm-cj} The secrecy outage
probability for the CJ scheme is given by
\begin{equation}
\mathcal{P}_{out}^{CJ}(R)=1-\frac{1}{\bar{\gamma}_{R,B}}\int_{t}^{\infty}e^{-\frac{2^{2R}-1}{\bar{\gamma}_{A,R}\phi(z)}-\frac{z}{\bar{\gamma}_{R,B}}}~dz,\label{eq:21}
\end{equation}
where
\begin{align}
 & \phi(z)=\frac{\rho z}{z+\bar{\gamma}_{A,R}+\bar{\gamma}_{R,B}+\frac{1}{\rho}}-\frac{2^{2R}}{z+\frac{1}{\rho}},\label{eq:22}\\
 & t=\frac{(2^{2R}-1)+\sqrt{(2^{2R}-1)^{2}+\rho2^{2R+1}(\bar{\gamma}_{A,R}+\bar{\gamma}_{R,B}+1/\rho)}}{2\rho}.\label{eq:23}
\end{align}
\end{proposition}
\begin{IEEEproof}
See Appendix \ref{sec:deriv-eqref}.
\end{IEEEproof}

It is important to note at this point that we have assumed each
node transmits at its maximum power budget $P$.  It is
straightforward to formulate a secrecy outage minimization
problem subject to various power constraints for the transmission
strategies discussed above, but obtaining closed-form solutions
to these problems appears to be intractable.  Thus our
theoretical analysis will focus on the case where all nodes
transmit with full power, but in the simulations presented later
we will show examples of the performance gain that can be obtained
with an optimal power allocation.

\subsection{Asymptotic Behavior}

\label{sec:asymptotic-behavior} Based on the above analytical expressions,
we see that the choice of which scheme (DT, AF or CJ) to employ depends
on the specific power budgets and channel gains; each of these methods
is optimal for different operating regimes. Next, we investigate the
asymptotic behavior of the outage probability to determine conditions
under which each approach offers the best performance.

\subsubsection{Case of $\rho\rightarrow\infty$}

\label{sec:case-p-rightarrow}

From \eqref{eq:9}, we have
\begin{equation}
\lim_{\rho\rightarrow\infty}\mathcal{P}_{out}^{DT}=1-\frac{\bar{\gamma}_{A,B}}{2^{R}\bar{\gamma}_{A,R}+\bar{\gamma}_{A,B}},\label{eq:24}
\end{equation}
and according to \eqref{eq:15},
\begin{equation}
\lim_{\rho\rightarrow\infty}\mathcal{P}_{out}^{AF}=1-\frac{\bar{\gamma}_{A,B}}{(2^{2R}-1)\bar{\gamma}_{A,R}+\bar{\gamma}_{A,B}}\left[\mu_{1}^{\prime}(\beta_{1}-1)e^{\mu_{1}^{\prime}\beta_{1}}\mathrm{Ei}(-\mu_{1}^{\prime}\beta_{1})+1\right]\label{eq:25}
\end{equation}
where $\mu_{1}^{\prime}=\frac{\bar{\gamma}_{A,R}}{\bar{\gamma}_{R,B}}$.
Therefore, both $\mathcal{P}_{out}^{DT}$ and $\mathcal{P}_{out}^{AF}$
converge to nonzero constants as $\rho\rightarrow\infty$. For CJ,
however, according to \eqref{eq:21}-\eqref{eq:23}, we have
\begin{equation}
\lim_{\rho\rightarrow\infty}\mathcal{P}_{out}^{CJ}=1-\frac{1}{\bar{\gamma}_{R,B}}\int_{0}^{\infty}e^{-\frac{z}{\bar{\gamma}_{R,B}}}~dz=0,\label{eq:26}
\end{equation}
which shows that CJ is preferable for $\rho\rightarrow\infty$, or
high SNR or transmit power scenarios.

\subsubsection{Case of $\bar{\gamma}_{A,B}\rightarrow\infty$ or $0$}

\label{sec:case-barg-right-1}

From \eqref{eq:9} and \eqref{eq:15}, it is also straightforward
to obtain that $\mathcal{P}_{out}^{DT},\mathcal{P}_{out}^{AF}\rightarrow0$
as $\bar{\gamma}_{A,B}\rightarrow\infty$. When $\bar{\gamma}_{A,B}$
is sufficiently large, using the fact that $1-e^{x}=x+O(x^{2})$,
we can observe that with respect to $\bar{\gamma}_{A,B}$, both DT
and AT decay proportionally to $1/\bar{\gamma}_{A,B}$. Conversely,
we also have that $\mathcal{P}_{out}^{DT},\mathcal{P}_{out}^{AF}\rightarrow1$
as $\bar{\gamma}_{A,B}\rightarrow0$. Since $\mathcal{P}_{out}^{CJ}$
does not depend on $\bar{\gamma}_{A,B}$, we can conclude that the
DT and AF schemes are better than CJ when the direct link is significantly
stronger than the others, while CJ will perform better when the direct
link is weak.

\subsubsection{Case of $\bar{\gamma}_{R,B}\rightarrow\infty$ or $0$}

\label{sec:case-barg-right}

Since $\mathcal{P}_{out}^{DT}$ is not a function of $\bar{\gamma}_{R,B}$,
it will be the same as in \eqref{eq:9}. When $\bar{\gamma}_{R,B}\rightarrow\infty$,
according to \eqref{eq:15} and using the result that $x\textrm{Ei}(-x)\rightarrow0$
as $x\rightarrow0$ \cite{Ding_Opportunistic11}, we have
\begin{equation}
\lim_{\bar{\gamma}_{R,B}\rightarrow\infty}\mathcal{P}_{out}^{AF}=1-\frac{\bar{\gamma}_{A,B}}{(2^{2R}-1)\bar{\gamma}_{A,R}+\bar{\gamma}_{A,B}}e^{-\frac{2^{2R}-1}{\rho\bar{\gamma}_{A,B}}}.\label{eq:27}
\end{equation}
For CJ when $\bar{\gamma}_{R,B}\rightarrow\infty$, intuitively the
cooperative jamming support from Bob can fully prevent the relay from
eavesdropping, and thus we will have the follwing lemma whose proof
is given in Appendix \ref{sec:proof-lemma-refs}.

\begin{lemma} \label{sec:case-bargamma_r-b} When $\bar{\gamma}_{R,B}\rightarrow\infty$,
outage events will only depend on the relay link, and the outage
probability in \eqref{eq:20} will converge to:
\begin{equation}
\lim_{\bar{\gamma}_{R,B}\rightarrow\infty}\mathcal{P}_{out}^{CJ}=\mathcal{P}\left(1+\rho\frac{|h_{R,B}|^{2}|h_{A,R}|^{2}}{|h_{R,B}|^{2}+\bar{\gamma}_{A,R}+\bar{\gamma}_{R,B}+\frac{1}{\rho}}<2^{2R}\right).\label{eq:28}
\end{equation}
\end{lemma}

Based on Lemma \ref{sec:case-bargamma_r-b} and following
the same approach as in Appendix \ref{sec:deriv-eqref}, we have
\begin{equation}
\lim_{\bar{\gamma}_{R,B}\rightarrow\infty}\mathcal{P}_{out}^{CJ}=\lim_{\bar{\gamma}_{R,B}\rightarrow\infty}1-\frac{1}{\bar{\gamma}_{R,B}}\int_{0}^{\infty}e^{-\frac{2^{2R}-1}{\bar{\gamma}_{A,R}h^{\prime}(z)}-\frac{z}{\bar{\gamma}_{R,B}}}~dz\label{eq:29}
\end{equation}
where
\[
h^{\prime}(z)=\frac{\rho z}{z+\bar{\gamma}_{A,R}+\bar{\gamma}_{R,B}+\frac{1}{\rho}}.
\]
Further manipulation of \eqref{eq:29} reveals
\begin{align}
\lim_{\bar{\gamma}_{R,B}\rightarrow\infty}\mathcal{P}_{out}^{CJ} & =\lim_{\bar{\gamma}_{R,B}\rightarrow\infty}1-\frac{1}{\bar{\gamma}_{R,B}}\int_{0}^{\infty}e^{-\frac{2^{2R}-1}{\bar{\gamma}_{A,R}h^{\prime}(z)}-\frac{z}{\bar{\gamma}_{R,B}}}~dz\notag\\
 & =\lim_{\bar{\gamma}_{R,B}\rightarrow\infty}1-\frac{1}{\bar{\gamma}_{R,B}}e^{-\frac{2^{2R}-1}{\rho\bar{\gamma}_{A,R}}}\int_{0}^{\infty}e^{-\frac{(2^{2R}-1)(\bar{\gamma}_{R,B}+\bar{\gamma}_{A,R}+1/\rho)}{\rho\bar{\gamma}_{A,R}z}-\frac{z}{\bar{\gamma}_{R,B}}}~dz\notag\\
 & =1-e^{-\frac{2^{2R}-1}{\rho\bar{\gamma}_{A,R}}}\sqrt{\frac{4(2^{2R}-1)}{\rho\bar{\gamma}_{A,R}}}\textrm{K}_{1}\left(\sqrt{\frac{4(2^{2R}-1)}{\rho\bar{\gamma}_{A,R}}}\right),\label{eq:30}
\end{align}
where $\textrm{K}_{1}(\cdot)$ is the modified Bessel function of
the second kind, and \cite[eq.~3.324.1]{Gradshteyn_Tables00} is used
to obtain \eqref{eq:30}. Therefore, as $\bar{\gamma}_{RB}\rightarrow\infty$,
the outage probability for all schemes converges to different constants,
and the analysis does not reveal an advantage of one method over another.

When $\bar{\gamma}_{RB}\rightarrow0$, CJ is obviously not applicable
since $\mathcal{P}_{out}^{CJ}\rightarrow1$. For the the AF scheme,
applying a procedure similar to that in Appendix \ref{sec:proof-lemma-refs}
on \eqref{eq:14}, the outage probability will converge to
\begin{align}
\lim_{\bar{\gamma}_{RB}\rightarrow0}\mathcal{P}_{out}^{AF} & =\mathcal{P}\left\{ \frac{1}{2}\log_{2}\left(\frac{1+\rho|h_{A,B}|^{2}}{1+\rho|h_{A,R}|^{2}}\right)<R\right\} \label{eq:31}\\
 & \ge\mathcal{P}\left\{ \log_{2}\left(\frac{1+\rho|h_{A,B}|^{2}}{1+\rho|h_{A,R}|^{2}}\right)<R\right\} \label{eq:32}
\end{align}
where \eqref{eq:32} is equal to $\mathcal{P}_{out}^{DT}$. Thus,
DT is a better choice when $\bar{\gamma}_{RB}\rightarrow0$ due to
the resource division factor $1/2$.

\subsubsection{Case of $\bar{\gamma}_{A,R}\rightarrow0$}

\label{sec:case-barg-right-2}

In this case, it is easy to verify from \eqref{eq:21} that, for CJ,
\begin{equation}
\lim_{\bar{\gamma}_{A,R}\rightarrow0}\mathcal{P}_{out}^{CJ}=1,\label{eq:33}
\end{equation}
since $t\rightarrow\infty$ as $\bar{\gamma}_{A,R}\rightarrow0$ and
the result of the integral in \eqref{eq:21} approaches $0$. However,
for DT and AF, the outage probability will converge to constants given
by
\begin{align}
\lim_{\bar{\gamma}_{A,R}\rightarrow0}\mathcal{P}_{out}^{DT} & =1-e^{-\frac{2^{R}-1}{\rho\bar{\gamma}_{A,B}}}\label{eq:34}\\
\lim_{\bar{\gamma}_{A,R}\rightarrow0}\mathcal{P}_{out}^{AF} & =1-e^{-\frac{2^{2R}-1}{\rho\bar{\gamma}_{A,B}}}.\label{eq:35}
\end{align}
Similar to the case in Section \ref{sec:case-barg-right}, we see
from the above equations that $\lim_{\bar{\gamma}_{A,R}\rightarrow0}\mathcal{P}_{out}^{AF}$
$\ge\lim_{\bar{\gamma}_{A,R}\rightarrow0}\mathcal{P}_{out}^{DT}$,
\emph{i.e.,} the SOP of the DT scheme is lower than that of the AF
scheme and thus DT is preferred in this case.

\subsection{Multi-antenna Relay }

\label{sec:multi-antenna-relay}

In this section, we generalize our analysis to the case of a multi-antenna relay.  We will theoretically
characterize the SOP and the impact of the number of relay
antennas on the secrecy performance.

\subsubsection{Direct Transmission (DT)}

Similar to the expression in \eqref{eq:9}, when the relay uses multiple
antennas and MRC beamforming, the outage probability for a given target secrecy rate $R$
is given by
\begin{equation}
\mathcal{P}_{out}^{DT}(R)=\mathcal{P}\left\{ \log_{2}\left(\frac{1+\rho|h_{A,B}|^{2}}{1+\rho\sum_{i=1}^{K}|h_{A,i}|^{2}}\right)<R\right\} .\label{eq:36}
\end{equation}
Denote $X=|h_{A,B}|^{2}$ and $Y=\sum_{i=1}^{K}|h_{A,i}|^{2}$, and
recall that $X\sim\exp\left(\frac{1}{\bar{\gamma}_{A,R}}\right)$.
Since $Y$ can be rewritten as $\frac{\bar{\gamma}_{A,R}}{2}\sum_{i=1}^{2K}\alpha_{i}^{2}$
where $\{\alpha\}_{i=1}^{2K}$ are $2K$ independent standard normal
Gaussian random variables (r.v.s), $\sum_{i=1}^{2K}\alpha_{i}^{2}$
has a central chi-square distribution with $2K$ degrees of freedom.
Therefore, we have
\begin{equation}
\mathcal{P}_{Y}(y)=\frac{y^{K-1}}{\bar{\gamma}_{A,R}^{K}(K-1)!}e^{-\frac{y}{\bar{\gamma}_{A,R}}}\label{eq:37}
\end{equation}
and thus
\begin{align}
\mathcal{P}_{out}^{DT} & =\mathbb{E}_{Y}\left\{ F_{X}\left(\frac{2^{R}-1}{\rho}+2^{R}Y\right)\right\} \notag\\
 & =1-\int_{0}^{\infty}\mathcal{P}_{Y}(y)~e^{-\frac{1}{\bar{\gamma}_{A,B}}\left(\frac{2^{R}-1}{\rho}+2^{R}Y\right)}~dy\notag\\
 & =1-\left(\frac{\bar{\gamma}_{A,B}}{2^{R}\bar{\gamma}_{A,R}+\bar{\gamma}_{A,B}}\right)^{K}e^{-\frac{2^{R}-1}{\rho\bar{\gamma}_{A,B}}}.\label{eq:39}
\end{align}
It can be seen from \eqref{eq:39} that the SOP of the DT scheme will
approach unity as $K$ grows, which is also consistent with the intuition
that the presence of more antennas at the eavesdropper will result in a deterioration in secrecy performance.

\subsubsection{Amplify-and-Forward (AF)}

According to the signal model \eqref{eq:1} and \eqref{eq:4}, the
outage probability in \eqref{eq:14} for a target secrecy rate $R$
and $K$ antennas at the relay can be written as
\begin{equation}
{P}_{out}^{AF}(R)=\mathcal{P}\left\{ \frac{1+\rho|h_{A,B}|^{2}+\rho\frac{\sum_{i=1}^{K}|h_{i,B}|^{2}\sum_{i=1}^{K}|h_{A,i}|^{2}}{\sum_{i=1}^{K}|h_{i,B}|^{2}+K\bar{\gamma}_{A,R}+\frac{K}{\rho}}}{1+\rho\sum_{i=1}^{K}|h_{A,i}|^{2}}<2^{2R}\right\} .\label{eq:40}
\end{equation}
The exact SOP is given by the following proposition:

\begin{proposition} \label{sec:amplify-forward-af} The secrecy outage
probability for AF relaying with $K$ relay antennas is
\begin{align}
 & \mathcal{P}_{out}^{AF}(R)=1-\left[1+\frac{\bar{\gamma}_{A,R}(2^{2R}-1)}{\bar{\gamma}_{A,B}}\right]^{-K}e^{-\frac{2^{2R}-1}{\rho\bar{\gamma}_{A,B}}}\notag\\*
 & +\int_{0}^{\infty}\int_{\frac{(1+\rho y)(2^{2R}-1)}{\rho}}^{2^{2R}y+\frac{2^{2R}-1}{\rho}}\frac{y^{K-1}}{\bar{\gamma}_{A,B}\bar{\gamma}_{A,R}^{K}(K-1)!}F_{V}\left(2^{2R}+\frac{2^{2R}-1-\rho x}{\rho y}\right)e^{-\frac{x}{\bar{\gamma}_{A,B}}-\frac{y}{\bar{\gamma}_{A,R}}}~dxdy,\label{eq:41}
\end{align}
where $R$ is the target secrecy rate, and the c.d.f. of $V$ is
\begin{equation}
F_{V}(v)=1-\sum_{n=0}^{K-1}\frac{1}{n!}\left[\frac{vK(\bar{\gamma}_{A,R}+\frac{1}{\rho})}{\bar{\gamma}_{R,B}(1-v)}\right]^{n}e^{-\frac{vK(\bar{\gamma}_{A,R}+\frac{1}{\rho})}{\bar{\gamma}_{R,B}(1-v)}}.
\end{equation}
\end{proposition}
\begin{IEEEproof}
Defining $X=|h_{A,B}|^{2}$, $Y=\sum_{i=1}^{K}|h_{A,i}|^{2}$, and
$V=\frac{\sum_{i=1}^{K}|h_{i,B}|^{2}}{\sum_{i=1}^{K}|h_{i,B}|^{2}+K\bar{\gamma}_{A,R}+K/\rho}$,
according to \eqref{eq:40} the SOP is given by
\begin{align}
\mathcal{P}_{out}^{AF}(R) & =\mathcal{P}\{1+\rho X-\rho(2^{2R}-V)Y<2^{2R}\}.\label{eq:42}
\end{align}
In order to obtain the SOP, we first compute the c.d.f. of $V$. It
is obvious that $F_{V}(v)=1$ when $v\ge1$ due to the fact that $V\le1$.
When $v<1$,
\begin{align}
F_{V}(v) & =\mathcal{P}\left\{ \frac{\sum_{i=1}^{K}|h_{i,B}|^{2}}{\sum_{i=1}^{K}|h_{i,B}|^{2}+K(\bar{\gamma}_{A,R}+\frac{1}{\rho})}\le v\right\} \notag\\
 & =\mathcal{P}\left\{ \sum_{i=1}^{K}|h_{i,B}|^{2}\le\frac{vK(\bar{\gamma}_{A,R}+1/\rho)}{1-v}\right\} \notag\\
 & \stackrel{a}{=}\mathcal{P}\left\{ \sum_{i=1}^{2K}\alpha_{i}^{2}\le2\frac{vK(\bar{\gamma}_{A,R}+1/\rho)}{\bar{\gamma}_{R,B}(1-v)}\right\} \notag\\
 & \stackrel{b}{=}1-\sum_{n=0}^{K-1}\frac{1}{n!}\left[\frac{vK(\bar{\gamma}_{A,R}+1/\rho)}{\bar{\gamma}_{R,B}(1-v)}\right]^{n}e^{-\frac{vK(\bar{\gamma}_{A,R}+1/\rho)}{\bar{\gamma}_{R,B}(1-v)}}\notag
\end{align}
where (b) results because $\{\alpha_{i}\}_{i=1}^{2K}$ are $2K$ independent
standard normal Gaussian r.v.s and thus $\sum_{i=1}^{2K}\alpha_{i}^{2}$
has a central chi-square distribution with $2K$ degrees of freedom.

Recall that the p.d.f.s for $X$ and $Y$ are $\mathcal{P}_{X}(x)=\frac{1}{\bar{\gamma}_{A,B}}e^{-\frac{x}{\bar{\gamma}_{A,B}}}$
and $\mathcal{P}_{Y}(y)=\frac{y^{K-1}}{\bar{\gamma}_{A,R}^{K}(K-1)!}e^{-\frac{y}{\bar{\gamma}_{A,R}}}$,
so the SOP can be written as
\begin{align}
 & \mathcal{P}_{out}^{AF}(R)=\mathbb{E}_{Y}\mathbb{E}_{X}\left\{ F_{V}\left(2^{2R}+\frac{2^{2R}-1-\rho X}{\rho Y}\right)\right\} \notag\label{eq:44}\\
 & =\int_{0}^{\infty}\int_{x_{l}}^{x_{u}}F_{V}\left(2^{2R}+\frac{2^{2R}-1-\rho X}{\rho Y}\right)\mathcal{P}_{X}(x)\mathcal{P}_{Y}(y)~dxdy+\int_{0}^{\infty}\int_{0}^{x_{l}}\mathcal{P}_{X}(x)\mathcal{P}_{Y}(y)~dxdy,
\end{align}
where the limits $x_{l}=\frac{(1+\rho y)(2^{2R}-1)}{\rho}$ and $x_{u}=2^{2R}y+\frac{2^{2R}-1}{\rho}$
can be derived from the fact that $0\le V\le1$. With some further
manipulations, the first term in \eqref{eq:44} can be computed as
\begin{align}
\int_{0}^{\infty}\int_{0}^{x_{l}}\mathcal{P}_{X}(x)\mathcal{P}_{Y}(y)~dxdy & =1-\int_{0}^{\infty}\frac{y^{K-1}}{\bar{\gamma}_{A,R}^{K}(K-1)!}e^{-y\left(\frac{1}{\bar{\gamma}_{A,R}}+\frac{2^{2R}-1}{\bar{\gamma}_{A,B}}\right)-\frac{2^{2R}-1}{\rho\bar{\gamma}_{A,B}}}\notag\label{eq:45}\\
 & =1-\left[1+\frac{\bar{\gamma}_{A,R}}{\bar{\gamma}_{A,B}}(2^{2R}-1)\right]^{-K}e^{-\frac{2^{2R}-1}{\rho\bar{\gamma}_{A,B}}}.
\end{align}
Inserting $\mathcal{P}_{X}(x)$, $\mathcal{P}_{Y}(y)$, $F_{V}(v)$
and \eqref{eq:45} to \eqref{eq:44}, the SOP for the multi-antenna
case is then obtained.
\end{IEEEproof}
\begin{corollary} \label{sec:amplify-forward-af-2} The secrecy outage
probability of AF relaying approaches unity as the number of relay
antennas grows: $\mathcal{P}_{out}^{AF}\rightarrow1$ as $K\rightarrow\infty$.
\end{corollary}
\begin{IEEEproof}
According to \eqref{eq:41} in Proposition \ref{sec:amplify-forward-af},
since $\left[1+\frac{\bar{\gamma}_{AR}(2^{2R}-1)}{\bar{\gamma}_{AB}}\right]^{-K}e^{-\frac{2^{2R}-1}{\rho\bar{\gamma}_{AB}}}$
converges to $0$ as $K$ grows, and since the third term in \eqref{eq:41}
is non-negative, the corollary follows in a straightforward manner.
\end{IEEEproof}

\subsubsection{Cooperative Jamming (CJ)}

According to \eqref{eq:5} and \eqref{eq:6}, the mutual information
between Alice and Bob can be expressed as
\begin{equation}
I_{B}^{CJ}=\frac{1}{2}\log_{2}\left(1+\rho\frac{\sum_{i=1}^{K}|h_{i,B}|^{2}\sum_{i=1}^{K}|h_{A,i}|^{2}}{\sum_{i=1}^{K}|h_{i,B}|^{2}+K\bar{\gamma}_{A,R}+K\bar{\gamma}_{R,B}+\frac{K}{\rho}}\right).\notag
\end{equation}
As discussed above, when computing the mutual information available to the relay, we assume the use of an optimal MMSE receive beamformer that allows the relay to maximize her SINR \cite{Huang_Cooperative11},
\textit{i.e.},
\begin{equation}
\mathbf{w}=\mathcal{G}\left(\mathbf{h}_{A,R}
\mathbf{h}_{A,R}^{H},\mathbf{h}_{R,B}\mathbf{h}_{R,B}^{H}+
\frac{1}{\rho}\mathbf{I}_{K}\right)\label{eq:47} \; ,
\end{equation}
where $\mathcal{G}(\mathbf{A},\mathbf{B})$ is the operator that returns
the eigenvector associated with the largest generalized eigenvalue
of the matrix pencil $(\mathbf{A},\mathbf{B})$. Since $\mathbf{h}_{A,R}\mathbf{h}_{A,R}^{H}$
is rank one, we can explicitly obtain the MMSE beamformer $\mathbf{w}=\frac{\left(\mathbf{h}_{R,B}\mathbf{h}_{R,B}^{H}+
\frac{1}{\rho}\mathbf{I}_{K}\right)^{-1}\mathbf{h}_{A,R}}
{||\left(\mathbf{h}_{R,B}\mathbf{h}_{R,B}^{H}+
\frac{1}{\rho}\mathbf{I}_{K}\right)^{-1}\mathbf{h}_{A,R}||}$,
and thus the mutual information between Alice and the relay is given by
\begin{equation}
I_{R}^{CJ}=\frac{1}{2}\log_{2}\left(1+\mathbf{h}_{A,R}^{H}\left(\mathbf{h}_{R,B}\mathbf{h}_{R,B}^{H}+\frac{1}{\rho}\mathbf{I}_{K}\right)^{-1}\mathbf{h}_{A,R}\right).\notag
\end{equation}
Therefore, the secrecy outage probability of the CJ scheme can
be written as
\begin{equation}
\mathcal{P}_{out}^{CJ}=\mathcal{P}\left\{ \frac{1+\rho\frac{\sum_{i=1}^{K}|h_{i,B}|^{2}\sum_{i=1}^{K}|h_{A,i}|^{2}}{\sum_{i=1}^{K}|h_{i,B}|^{2}+K\bar{\gamma}_{A,R}+K\bar{\gamma}_{R,B}+\frac{K}{\rho}}}{1+\mathbf{h}_{A,R}^{H}\left(\mathbf{h}_{R,B}\mathbf{h}_{R,B}^{H}+\frac{1}{\rho}\mathbf{I}_{K}\right)^{-1}\mathbf{h}_{A,R}}<2^{2R}\right\} .\label{eq:49}
\end{equation}

Unlike the analysis for the AF scheme, \eqref{eq:49} is more complicated
and it is unclear how to obtain a closed-form SOP expression. Instead,
we focus here on finding an asymptotic SOP with respect to the transmit
SNR and the number of relay antennas, as detailed in the following
corollaries.

\begin{corollary} When $K>1$, the secrecy outage probability of
the CJ scheme approaches a constant as $\rho\rightarrow\infty$. \end{corollary}
\begin{IEEEproof}
When $\rho\rightarrow\infty$, the receive beamformer in \eqref{eq:47}
converges to $\mathbf{w}\in\mathcal{N}(\mathbf{h}_{R,B})$ where $\mathcal{N}(\cdot)$
represents the null space operator. As a result, $I_{R}^{CJ}\rightarrow\frac{1}{2}\log_{2}\left(1+\rho\mathbf{h}_{A,R}^{H}\mathbf{h}_{A,R}\right)$
and thus we have
\begin{equation}
\mathcal{P}_{out}^{CJ}\simeq\mathcal{P}\left\{ \frac{\frac{\sum_{i=1}^{K}|h_{i,B}|^{2}\sum_{i=1}^{K}|h_{A,i}|^{2}}{\sum_{i=1}^{K}|h_{i,B}|^{2}+K\bar{\gamma}_{A,R}+K\bar{\gamma}_{R,B}}}{\sum_{i=1}^{K}|h_{A,i}|^{2}}<2^{2R}\right\} \label{eq:50}
\end{equation}
which is a nonzero constant independent of $\rho$.
\end{IEEEproof}
\indent

Recall from the previous section that unlike the above corollary,
when $K=1$ the SOP of CJ converges to zero, which indicates a more
favorable scenario for using CJ. Clearly, this is due to the fact
that a relay with multiple antennas is able to suppress the jamming
signal from Bob. This point is formalized in the following corollary.

\begin{corollary} \label{sec:coop-jamm-cj-1} The secrecy outage
probability of the CJ scheme approaches unity as the number of relay
antennas grows: $\mathcal{P}_{out}^{CJ}\rightarrow1$ as $K\rightarrow\infty$.
\end{corollary}
\begin{IEEEproof}
Using the Sherman-Morrison formula, we have
\begin{equation}
\left(\mathbf{h}_{R,B}\mathbf{h}_{R,B}^{H}+\frac{1}{\rho}\mathbf{I}_{K}\right)^{-1}=\rho\mathbf{I}_{K}-\rho^{2}\frac{\mathbf{h}_{R,B}\mathbf{h}_{R,B}^{H}}{1+\rho||\mathbf{h}_{R,B}||^{2}}\label{eq:51}
\end{equation}
and thus \eqref{eq:49} can be rewritten as
\begin{align}
\mathcal{P}_{out}^{CJ} & =\mathcal{P}\left\{ \frac{1+\rho\frac{\sum_{i=1}^{K}|h_{i,B}|^{2}\sum_{i=1}^{K}|h_{A,i}|^{2}}{\sum_{i=1}^{K}|h_{i,B}|^{2}+K\bar{\gamma}_{A,R}+K\bar{\gamma}_{R,B}+\frac{K}{\rho}}}{1-\rho^{2}\mathbf{h}_{A,R}\mathbf{G}\mathbf{h}_{A,R}^{H}+\rho\sum_{i=1}^{K}|h_{A,i}|^{2}}<2^{2R}\right\} \notag\label{eq:52}\\*
 & \stackrel{a}{\ge}\mathcal{P}\left\{ \frac{1+\rho\sum_{i=1}^{K}|h_{A,i}|^{2}}{1-\rho^{2}\mathbf{h}_{A,R}\mathbf{G}\mathbf{h}_{A,R}^{H}+\rho\sum_{i=1}^{K}|h_{A,i}|^{2}}<2^{2R}\right\}
\end{align}
where $\mathbf{G}=\frac{\mathbf{h}_{R,B}\mathbf{h}_{R,B}^{H}}{1+\rho||\mathbf{h}_{R,B}||^{2}}$
and inequality (a) holds since $\frac{\sum_{i=1}^{K}|h_{i,B}|^{2}}{\sum_{i=1}^{K}|h_{i,B}|^{2}+K\bar{\gamma}_{A,R}+K\bar{\gamma}_{R,B}+\frac{K}{\rho}}\le1$.
Since $\mathbf{G}$ is a unit-rank Hermitian matrix, it can be seen
that $\mathbf{h}_{A,R}\mathbf{G}\mathbf{h}_{A,R}^{H}$ is exponentially
distributed as $\exp\left(\frac{1}{\bar{\gamma}_{A,R}\lambda_{G}}\right)$
where $\lambda_{G}$ is the largest eigenvalue of $\mathbf{G}$ and
$\lim_{K\rightarrow\infty}\lambda_{G}=\frac{1}{\rho}$. Consequently,
$\mathbf{h}_{A,R}\mathbf{G}\mathbf{h}_{A,R}^{H}$ is not a function
of $K$ as $K\rightarrow\infty$. On the other hand, $\sum_{i=1}^{K}|h_{A,i}|^{2}$
is obviously an increasing function of $K$ and $\sum_{i=1}^{K}|h_{A,i}|^{2}\rightarrow\infty$
as $K\rightarrow\infty$. Therefore, the lower bound in \eqref{eq:52}
approaches unity and the proof is complete.
\end{IEEEproof}
\indent

Corollarys \ref{sec:amplify-forward-af-2} and \ref{sec:coop-jamm-cj-1}
provide pessimistic conclusions regarding secrecy for an untrusted
relay implementing beamforming with a large number of antennas. However,
in the next section, we show that if the relay is
forced to perform antenna selection (e.g., because it only has a single
RF chain), under certain conditions an increase in the number of relay
antennas actually improves secrecy. It is also worth noting that in
this paper we adopt a relaying protocol where Alice does not transmit
information signals to Bob in the second phase, and thus the conclusions obtained above may not hold for other relaying protocols, such as those
where Alice can transmit signals in the second phase (\emph{e.g.}
Protocol III in \cite{Nabar_Fading04}).

\section{Secrecy with Relay Antenna Selection}

\label{sec:secrecy-with-antenna}

In this section, we consider a scenario where the untrusted relay
must perform antenna selection for receive and transmit, rather than
beamforming. In particular, we assume the untrusted relay chooses
the receive antenna with the largest channel gain for maximizing her
wiretapping ability in the first hop, while still assisting Alice
by using the best transmit antenna to forward the message to Bob in
the second hop. Such behavior is consistent with a relay that is untrusted
but not malicious. This follows a similar CSI-based antenna selection
approach assumed in traditional relaying systems \cite{Krikidis_Relay09,Amarasuriya_Feedback10,Kim_End06}.
Since the relay loses the flexibility of using beamforming to cope
with the artificial jamming signals, and since Bob is still able to
enjoy a diversity benefit due to antenna selection, a secrecy performance
improvement is expected for the CJ scheme. We will characterize the
SOP for AF and CJ and study the impact of the number of relay antennas.
We let $m$ and $n$ respectively denote the indices of the receive
and transmit antenna used by the relay.

\subsection{Direct Transmission (DT)}

In this scheme, the untrusted relay chooses the best antenna to wiretap
the signal from Alice, and the resulting SOP can be found by a straightforward
extension of \eqref{eq:36}, which in this case becomes
\begin{equation}
\mathcal{P}_{out}^{DT}(R)=\mathcal{P}\left\{ \log_{2}\left(\frac{1+\rho|h_{A,B}|^{2}}{1+\rho|h_{A,m^{*}}|^{2}}\right)<R\right\} ,\label{eq:53}
\end{equation}
where $m^{*}=\arg\max_{m}\{|h_{A,m}|^{2}\}$. Assume $Z=|h_{A,B}|^{2}$,
$Y=|h_{A,m}|^{2}$ and $Y^{*}=|h_{A,m^{*}}|^{2}$. Since $|h_{A,m^{*}}|^{2}$
increases as $K$ grows, it is obvious that $\mathcal{P}_{out}^{DT}(R)\rightarrow1$
as $K\rightarrow\infty$. To obtain the exact SOP, since $Z$ and
$Y$ are both exponentially distributed, using the theory of order
statistics \cite{Papoulis_Probability02} we have
\begin{equation}
p_{Y^{*}}(y)=\frac{K}{\bar{\gamma}_{A,R}}\sum_{n=0}^{K-1}\binom{K-1}{n}(-1)^{n}e^{-\frac{y}{\bar{\gamma}_{A,R}}(n+1)}.\label{eq:54}
\end{equation}
Therefore, the SOP can be computed as
\begin{align}
\mathcal{P}_{out}^{DT} & =\mathbb{E}_{Y}^{*}\left\{ F_{Z}\left(\frac{2^{R}-1}{\rho}+2^{R}Y^{*}\right)\right\} \notag\\
 & =1-K\sum_{n=0}^{K-1}\binom{K-1}{n}(-1)^{n}\frac{\bar{\gamma}_{A,B}}{2^{R}\bar{\gamma}_{A,R}+\bar{\gamma}_{A,B}(n+1)}e^{-\frac{2^{R}-1}{\rho\bar{\gamma}_{A,B}}}.\label{eq:55}
\end{align}

\subsection{Amplify-and-Forward (AF)}

Here we consider two cases, one where the relay has Bob's CSI for
the second hop, and one where it does not. The latter case corresponds
to the scenario where Bob is a passive receiver or where he simply
does not transmit training data to the relay.

\subsubsection{Relaying with second-hop CSI}

Similar to \eqref{eq:14}, the SOP of the AF scheme with antenna selection
for a given secrecy rate $R$ is given by
\begin{equation}
\mathcal{P}_{out}^{AF}(R)=\mathcal{P}\left\{ \frac{1}{2}\log_{2}\left(\frac{1+\rho|h_{A,B}|^{2}+\rho\frac{|h_{n^{*},B}|^{2}|h_{A,m^{*}}|^{2}}{|h_{n^{*},B}|^{2}+\bar{\gamma}_{A,R}+\frac{1}{\rho}}}{1+\rho|h_{A,m^{*}}|^{2}}\right)<R\right\} \notag,
\end{equation}
where the receive and transmit antennas on the relay are selected
using the following criteria:
\begin{align}
m^{*} & =\arg\max_{m}\{|h_{A,m}|^{2}\}\label{eq:57}\\
n^{*} & =\arg\max_{n}\{|h_{n,B}|^{2}\} \; . \label{eq:58}
\end{align}
These criteria are obtained assuming that the
untrusted relay will maximize her SNR for wiretapping first with \eqref{eq:57}
and then consider offering assistance to Bob with \eqref{eq:58}. An exact
expression for the SOP in this case is given by the following proposition.

\begin{proposition} \label{sec:amplify-forward-af2} The secrecy
outage probability for AF relaying with antenna selection can be expressed
as
\begin{align}
 \mathcal{P}_{out}^{AF}(R)= & 1-K^{2}\sum_{m=1}^{K}\sum_{n=1}^{K}\binom{K-1}{m}\binom{K-1}{n}(-1)^{m+n}\notag\label{eq:59}\\*
 & \times\frac{\bar{\gamma}_{A,B}}{(2^{2R}-1)\bar{\gamma}_{A,R}+\bar{\gamma}_{A,B}(n+1)}e^{-\frac{2^{2R}-1}{\rho\bar{\gamma}_{A,B}}}\notag\\*
 & \times\left[\mu(\beta_{n}-1)e^{\mu\beta_{n}(m+1)}\mathrm{Ei}(-\mu\beta_{n}(m+1))+\frac{1}{m+1}\right]
\end{align}
where $\mu=\frac{\bar{\gamma}_{A,R}+1/\rho}{\bar{\gamma}_{R,B}}$,
$\beta_{n}=\frac{2^{2R}\bar{\gamma}_{A,R}+\bar{\gamma}_{A,B}(n+1)}{(2^{2R}-1)\bar{\gamma}_{A,R}+\bar{\gamma}_{A,B}(n+1)}$,
$R$ is the target secrecy rate, and $\mathrm{Ei}(\cdot)$ is the
exponential integral $\mathrm{Ei}(x)=\int_{-\infty}^{x}e^{t}t^{-1}~dt$.
\end{proposition}
\begin{proof} Define $X=|h_{A,B}|^{2}$, $Y=|h_{A,m}|^{2}$,
$V=\frac{|h_{n,B}|^{2}}{|h_{n,B}|^{2}+\bar{\gamma}_{A,R}+1/\rho}$,
$Y^{*}=|h_{A,m^{*}}|^{2}$, $V^{*}=\frac{|h_{n^{*},B}|^{2}}{|h_{n^{*},B}|^{2}+\bar{\gamma}_{A,R}+1/\rho}$,
and note that the p.d.f. of $Y^{*}$ is given in \eqref{eq:54}. For V, using
the Jacobian transformation, we have
\[
p_{V}(v)=\frac{\bar{\gamma}_{A,R}+1/\rho}{\bar{\gamma}_{R,B}(1-v)^{2}}e^{-\frac{(\bar{\gamma}_{A,R}+1/\rho)v}{\bar{\gamma}_{R,B}(1-v)}},
\]
and the p.d.f. of $V^{*}$ can be expressed using order statistics
as
\begin{equation}
p_{V^{*}}(v)=\frac{K\mu}{(1-v)^{2}}\sum_{m=0}^{K-1}\binom{K-1}{m}(-1)^{m}e^{-\frac{\mu v}{1-v}(m+1)},\label{eq:60}
\end{equation}
where $\mu=\frac{\bar{\gamma}_{A,R}+1/\rho}{\bar{\gamma}_{R,B}}$.
The proof of \eqref{eq:59} is completed by inserting \eqref{eq:54}
and \eqref{eq:60} into
\[
{P}_{out}^{AF}(R)=\mathcal{P}\left\{ Z<2^{2R}\right\} =\mathbb{E}_{V^{*}}\{\mathbb{E}_{Y^{*}}\{F_{Z|Y^{*},V^{*}}(2^{2R})\}\}
\]
where $Z=\frac{1+\rho X+\rho V^{*}Y^{*}}{1+\rho Y^{*}}$. \end{proof}

\begin{corollary} \label{sec:amplify-forward-af-1} The secrecy outage
probability of AF relaying approaches unity as the number of relay
antennas grows: $\mathcal{P}_{out}^{AF}\rightarrow1$ as $K\rightarrow\infty$.
\end{corollary}
\begin{IEEEproof}
This corollary can be proved by showing that a lower bound for $\mathcal{P}_{out}^{AF}$
goes to 1 as $K\rightarrow\infty$. Following the notation in the
proof of Proposition \ref{sec:amplify-forward-af2}, we have
\begin{align}
\mathcal{P}_{out}^{AF}(R) & =\mathcal{P}\left(\frac{1+\rho X+\rho V^{*}Y^{*}}{1+\rho Y^{*}}<2^{2R}\right)\notag\\*
 & \stackrel{a}{\ge}\mathcal{P}\left(\frac{1+\rho X+\rho Y^{*}}{1+\rho Y^{*}}<2^{2R}\right)\label{eq:61}\\*
 & \stackrel{b}{\ge}\mathcal{P}\left(\frac{X}{Y^{*}}<2^{2R}-1\right)\notag\\
 & =\mathcal{P}\left(\min_{m}\left\{ \frac{|h_{A,B}|^{2}}{|h_{A,m}|^{2}}\right\} <2^{2R}-1\right)\notag\\*
 & \stackrel{c}{=}1-\left[1-\frac{\bar{\gamma}_{A,R}(2^{2R}-1)}{\bar{\gamma}_{A,B}+\bar{\gamma}_{A,R}(2^{2R}-1)}\right]^{K},\label{eq:62}
\end{align}
where it is obvious that \eqref{eq:62} converges to 1 as $K$ goes
to $\infty$. Inequality (a) holds since $V^{*}\le1$. The fraction
in \eqref{eq:61} is a quasi-linear function of $\rho$, and is monotonically
increasing with respect to $\rho$ since $X+Y^{*}\ge Y^{*}$; thus
inequality (b) is obtained by letting $\rho\rightarrow\infty$. To
obtain (c), we have used the result in \eqref{eq:70} that
\[
\mathcal{P}\left\{ \frac{|h_{A,B}|^{2}}{|h_{A,m}|^{2}}\le u\right\} =\frac{\bar{\gamma}_{A,R}u}{\bar{\gamma}_{A,B}+\bar{\gamma}_{A,R}u}.
\]

\end{IEEEproof}
Corollary \ref{sec:amplify-forward-af-1} shows that although both
the relay and Bob receive diversity gain from an increasing number
of relay antennas, the untrusted relay accrues a proportionally greater
benefit to the detriment of the information confidentiality.

\subsubsection{Relaying without second-hop CSI}

In this case, the relay is forced to choose a random antenna for the
second hop transmission, and the exact SOP is simply a special case
of \eqref{eq:59}:
\begin{align}
{P}_{out}^{AF}(R) & =\mathbb{E}_{V}\{\mathbb{E}_{Y^{*}}\{F_{Z|Y^{*},V}(2^{2R})\}\}\notag\label{eq:63}\\
 & =1-K\sum_{n=1}^{K}\binom{K-1}{n}(-1)^{n}\frac{\bar{\gamma}_{A,B}[\mu(\beta_{n}-1)e^{\mu\beta_{n}}\mathrm{Ei}(-\mu\beta_{n})+1]}{(2^{2R}-1)\bar{\gamma}_{A,R}+\bar{\gamma}_{A,B}(n+1)}e^{-\frac{2^{2R}-1}{\rho\bar{\gamma}_{A,B}}}
\end{align}
where $\mu$ and $\beta_{n}$ were given in Proposition \ref{sec:amplify-forward-af2}.
It is obvious that the performance in this case is always worse than
the case where the second-hop CSI is available. However, we will see below that
for CJ, the lack of second-hop CSI can lead to improved secrecy.

\subsection{Cooperative Jamming (CJ)}

\label{sec:coop-jamm-cj-3}

The relay's antenna selection protocol is slightly different in this
case since the relay must account for the interference from Bob in
the first hop.

\subsubsection{Relaying with second-hop CSI}

We first consider the case where the relay possesses the CSI for the
second-hop. Similar to \eqref{eq:20}, the corresponding SOP for the
CJ protocol is given by
\begin{equation}
\mathcal{P}_{out}^{CJ}(R)=\mathcal{P}\left(\frac{1+\rho\frac{|h_{n^{*},B}|^{2}|h_{A,m^{*}}|^{2}}{|h_{n^{*},B}|^{2}+\bar{\gamma}_{A,R}+\bar{\gamma}_{R,B}+\frac{1}{\rho}}}{1+\frac{|h_{A,m^{*}}|^{2}}{|h_{m^{*},B}|^{2}+\frac{1}{\rho}}}<2^{2R}\right)\label{eq:64}
\end{equation}
but in this case the receive and transmit antenna at the relay are
selected by
\begin{align}
m^{*} & =\arg\max_{m}\left\{ \frac{|h_{A,m}|^{2}}{|h_{m,B}|^{2}}\right\} \label{eq:65}\\
n^{*} & =\arg\max_{n}\{|h_{n,B}|^{2}\}.
\end{align}
where \eqref{eq:65} indicates that to improve its performance, the
relay chooses its receive antenna to maximize the ratio of the power
of Alice's signal to the power of Bob's jamming. It is difficult to
exactly calculate the SOP in this case. However, it can still be observed
that $\mathcal{P}_{out}^{CJ}\rightarrow1$ as $K\rightarrow\infty$
since the denominator in \eqref{eq:64} tends to increase with the
growth of $K$, while the numerator in \eqref{eq:64} is upper bounded
by $1+\rho|h_{A,m^{*}}|^{2}$ which will not necessarily increase
as $K$ grows. However, as explained next, a different conclusion
is obtained if the relay does not possess Bob's CSI.

\subsubsection{Relaying without second-hop CSI}

Here we assume that the relay has no information about $h_{R,B}$,
which applies to the case where Bob transmits no training data to
the relay, and jams only when Alice is transmitting so the relay cannot
collect interference information. Alternatively, a more advanced training
sequence design, such as the methods for discriminatory channel
estimation in \cite{Chang_Training10,Huang_Two11}, can be applied
to prevent the relay from acquiring CSI from Bob. Thus, the relay
uses the receive antenna with the largest channel gain during the
first hop, and then uses either the same or some random antenna for
transmission during the second hop. In this case, the SOP is given
by a slightly different expression than \eqref{eq:64}:
\begin{equation}
\mathcal{P}_{out}^{CJ}(R)=\mathcal{P}\left(\frac{1+\rho\frac{|h_{m^{*},B}|^{2}|h_{A,m^{*}}|^{2}}{|h_{m^{*},B}|^{2}+\bar{\gamma}_{A,R}+\bar{\gamma}_{R,B}+\frac{1}{\rho}}}{1+\frac{|h_{A,m^{*}}|^{2}}{|h_{m^{*},B}|^{2}+\frac{1}{\rho}}}<2^{2R}\right)\label{eq:66}
\end{equation}
where $m^{*}=\arg\max_{m}\{|h_{A,m}|^{2}\}$. Note that since $|h_{m^{*},B}|^{2}$
is independent of $|h_{A,m^{*}}|^{2}$, it is equivalent to selecting
the transmit antenna randomly. Therefore, following the result in
Appendix \ref{sec:deriv-eqref}, we can compute the SOP as
\begin{align}
\mathcal{P}_{out}^{CJ}(R) & =\mathcal{P}\left(\phi(|h_{m^{*},B}|^{2})|h_{A,m^{*}}|^{2}<2^{2R}-1\right)\notag\\
 & =\frac{1}{\bar{\gamma}_{R,B}}\int_{0}^{t}e^{-\frac{z}{\bar{\gamma}_{R,B}}}~dz+\frac{1}{\bar{\gamma}_{R,B}}\int_{t}^{\infty}\left[1-e^{-\frac{2^{2R}-1}{\bar{\gamma}_{A,R}\phi(z)}}\right]^{K}e^{-\frac{z}{\bar{\gamma}_{R,B}}}~dz\label{eq:67}\\
 & =(1-e^{-\frac{t}{\bar{\gamma}_{R,B}}})+\frac{1}{\bar{\gamma}_{R,B}}\sum_{n=0}^{K}\binom{K}{n}(-1)^{n}\int_{t}^{\infty}e^{-\frac{(2^{2R}-1)n}{\bar{\gamma}_{A,R}\phi(z)}-\frac{z}{\bar{\gamma}_{R,B}}}~dz,\label{eq:68}
\end{align}
where $\phi(z)$ and $t$ are provided in \eqref{eq:79} and \eqref{eq:82}
respectively. According to \eqref{eq:67} and \eqref{eq:68}, we can
also give the following corollary.

\begin{corollary} \label{sec:relay-with-second} Without second-hop
CSI at the relay, the SOP of the CJ scheme with antenna selection
decreases as $K$ grows and converges to $1-e^{-\frac{t}{\bar{\gamma}_{R,B}}}$,
where
\[
t=\frac{(2^{2R}-1)+\sqrt{(2^{2R}-1)^{2}+\rho2^{2R+1}(\bar{\gamma}_{A,R}+\bar{\gamma}_{R,B}+1/\rho)}}{2\rho}.
\]
\end{corollary}

This corollary indicates that if the second-hop CSI can be hidden
from the relay, although the legitimate user loses the diversity benefits
that come from transmit antenna selection, the overall secrecy
performance is still improved since the relay loses the diversity
gain due to cooperative jamming, while the legitimate user can
still achieve a diversity gain from the first-hop antenna selection.

\section{Numerical Results}

\label{sec:nr}

In this section, we present numerical examples of the outage performance for the DT, AF and CJ transmission schemes in both single-antenna
and multi-antenna scenarios. The SOP is computed for various values of the transmit powers,
average channel gains, and number of antennas. In all cases, the normalized
target secrecy rate is set to $R=0.1$ bits per channel use as assumed
in \cite{Bloch_Wireless08,Krikidis_Relay09}.

\subsection{Single-Antenna Case\label{sec:single-relay-case}}

\begin{figure}[h]
\centering \includegraphics[width=0.55\textwidth]{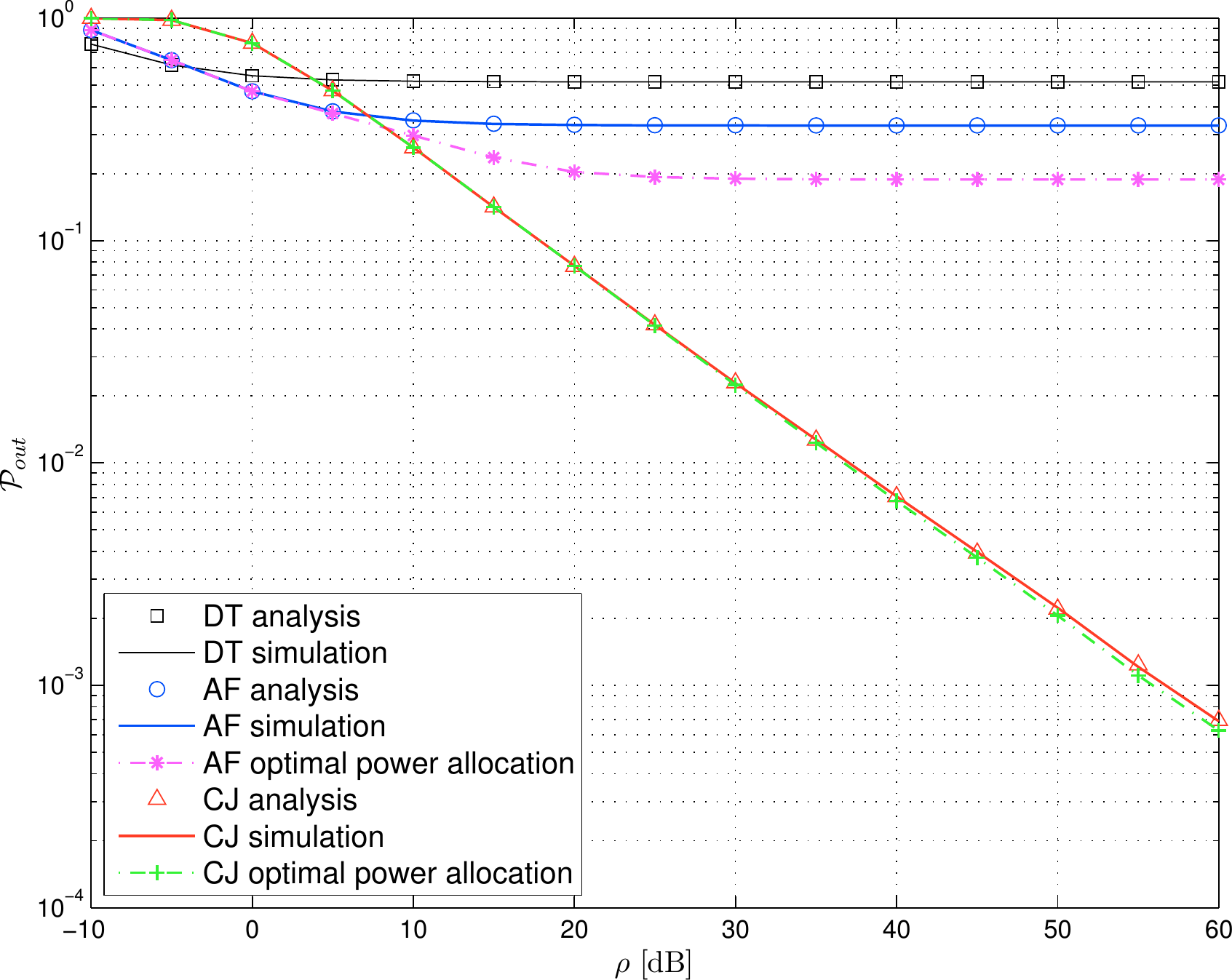}
\caption{\label{fig:osp}Outage probability versus $\rho$, single antenna
relay, $\bar{\gamma}_{A,B}=\bar{\gamma}_{A,R}=0$dB, $\bar{\gamma}_{R,B}=5$dB,
analytical results computed with Eqs. \eqref{eq:9} for DT, \eqref{eq:15}
for AF, and \eqref{eq:21} for CJ.}
\end{figure}

Fig.~\ref{fig:osp} depicts the outage probability as a function
of the transmit SNR $\rho$, assuming the average channel gains are
$\bar{\gamma}_{A,B}=\bar{\gamma}_{A,R}=0$dB, $\bar{\gamma}_{R,B}=5$dB.
The analytical SOP results for DT, AF and CJ are evaluated through
Eqs.~\eqref{eq:9}, \eqref{eq:15} and \eqref{eq:21}, and are seen
to agree well with the simulations, exactly predicting the performance
cross-over points. Ignoring the available relay link and treating
it as a pure adversary as in DT is clearly suboptimal for medium to
high SNR regimes. This figure shows that when $\rho\rightarrow\infty$,
the outage probability converges to a constant for DT and AF while
it goes to 0 for CJ, which agrees with the discussion in Section \ref{sec:case-p-rightarrow}.
This is due to the fact that the jamming signals from Bob only selectively
interfere with the untrusted relay and have no impact on the overall
two-hop data signal reception. Therefore, the outage performance for
CJ is better than AF for high SNR, while the converse is true
in the low SNR regime. We also show in this figure the SOPs
for AF and CJ assuming an optimal power allocation obtained by direct
numerical optimization.  It can be seen that the performance gap between the fixed and optimal power allocations is more obvious for AF than for the CJ scheme, since AF utilizes the direct link between Alice
and Bob and thus the secrecy performance is more sensitive to the
power allocation. The figure also illustrates that the performance under the optimal power allocation still follows the asymptotic analysis conducted in Section~\ref{sec:asymptotic-behavior}.

The impact of $\bar{\gamma}_{R,B}$ on performance is illustrated
in Fig.~\ref{fig:osrrb}, where $\bar{\gamma}_{A,B}=5$dB, $\bar{\gamma}_{A,R}=0$dB,
and $\rho=15$dB. Observe that when $\bar{\gamma}_{R,B}\rightarrow0$,
the outage probability for CJ approaches 1, due to its sensitivity
to the quality of the second hop, while the performance of both DT
and AF converges to nonzero constant values. Although not obvious,
DT still exhibits a gain over AF due to its efficient resource usage,
as can be seen from Eqs.~\eqref{eq:31} and \eqref{eq:32}. Note
that from \eqref{eq:32}, we expect that this gain will increase with
a higher target secrecy rate $R$. It is also worth noting that although
this figure shows that CJ has the best performance as $\bar{\gamma}_{R,B}\rightarrow\infty$,
the relative performance of these schemes will change with different
values of $\bar{\gamma}_{A,R}$ and $\bar{\gamma}_{A,B}$, and thus
we can not draw any definite conclusions in this asymptotic case.
\begin{figure}[h]
\centering \includegraphics[width=0.55\textwidth]{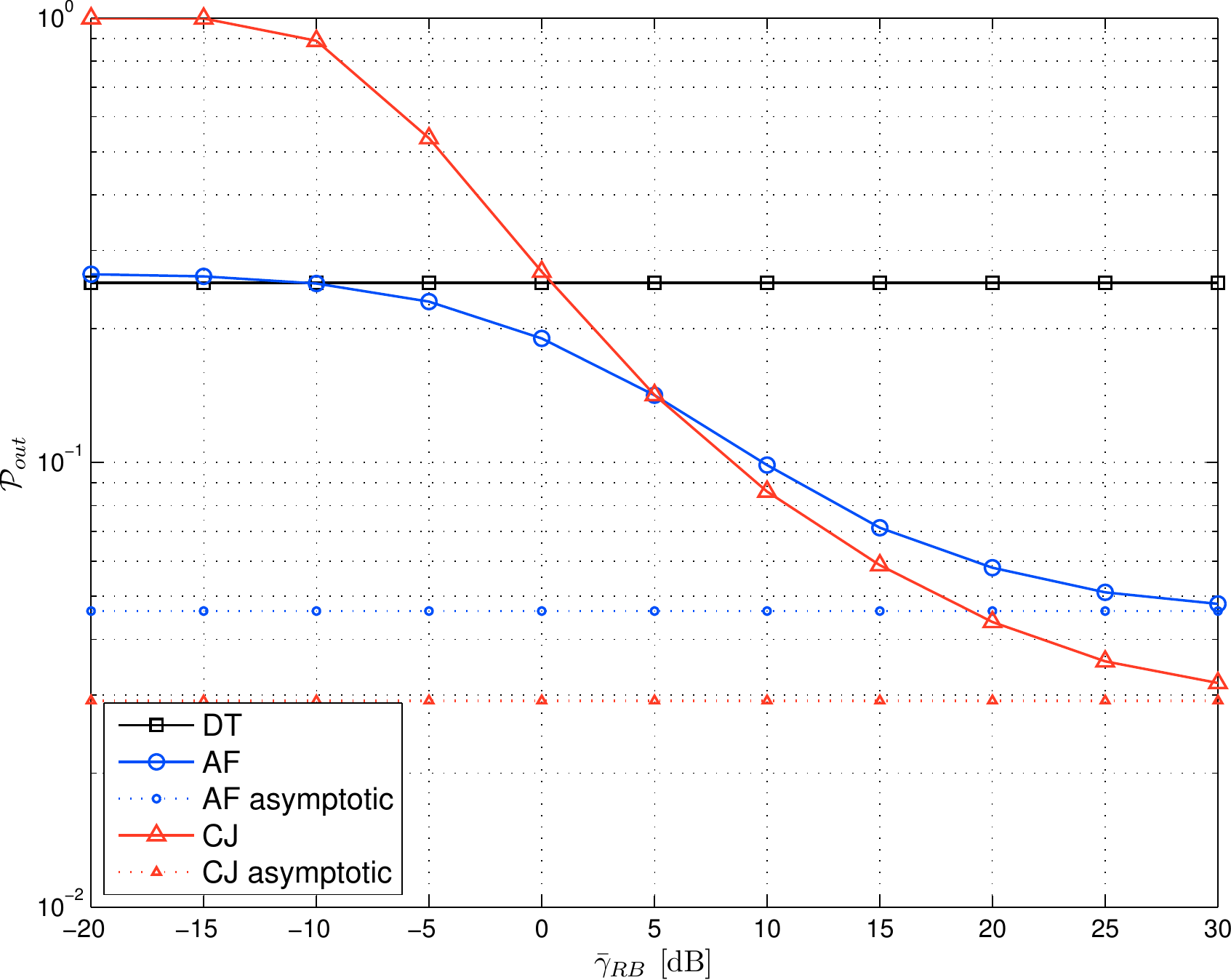}
\caption{\label{fig:osrrb} Outage probability versus $\bar{\gamma}_{R,B}$,
single-antenna relay, $\bar{\gamma}_{A,B}=5$dB, $\bar{\gamma}_{A,R}=0$dB,
$\rho=15$dB, asymptotic results computed with Eqs.~\eqref{eq:27}
for AF, and \eqref{eq:30} for CJ.}
\end{figure}

Fig.~\ref{fig:osrar} depicts the impact of the first hop channel
gain $\bar{\gamma}_{A,R}$ on the outage performance, where $\bar{\gamma}_{A,B}=0$dB,
$\bar{\gamma}_{A,R}=5$dB, and $\rho=20$dB. It is interesting to
see that when $\bar{\gamma}_{A,R}$ is either extremely small or large,
CJ approaches outage. This is because when $\bar{\gamma}_{A,R}\rightarrow\infty$,
the untrusted relay is nearly colocated with Alice and secure transmission
is impossible. On the other hand, when $\bar{\gamma}_{A,R}\rightarrow0$,
it is hard to establish a reliable relay link from Alice to Bob without
the direct link, and thus the outage probability will also approach
unity. Therefore, CJ can only achieve its best performance for in-between
values of $\bar{\gamma}_{A,R}$. Again, our analysis allows the optimal
operating regime for CJ to be determined. Also, as $\bar{\gamma}_{A,R}\rightarrow0$,
the asymptotic results validate the analytical expectations in \eqref{eq:34}
and \eqref{eq:35} which predict that DT will asymptotically outperform
AF. Therefore, the outage performance in Fig.~\ref{fig:osrrb} and
Fig.~\ref{fig:osrar} agrees with the analytical prediction in Section
\ref{sec:case-p-rightarrow} that DT is preferred when either relay
hop is weak.

Fig.~\ref{fig:osrab} shows the performance as a function of $\bar{\gamma}_{A,B}$,
with $\bar{\gamma}_{A,R}=2$dB, $\bar{\gamma}_{R,B}=10$dB and $\rho=10$dB.
It is shown that when $\bar{\gamma}_{A,B}$ is small, CJ is the best
scheme since both DT and AF will be in outage. Conversely, with large
$\bar{\gamma}_{A,B}$, the outage probability for DT and AF decays
to 0. Moreover, as discussed in Section \ref{sec:case-barg-right-1},
the outage probability for both DT and AF is seen to decay as $1/\bar{\gamma}_{A,B}$.

\begin{figure}[h]
\centering \includegraphics[width=0.55\textwidth]{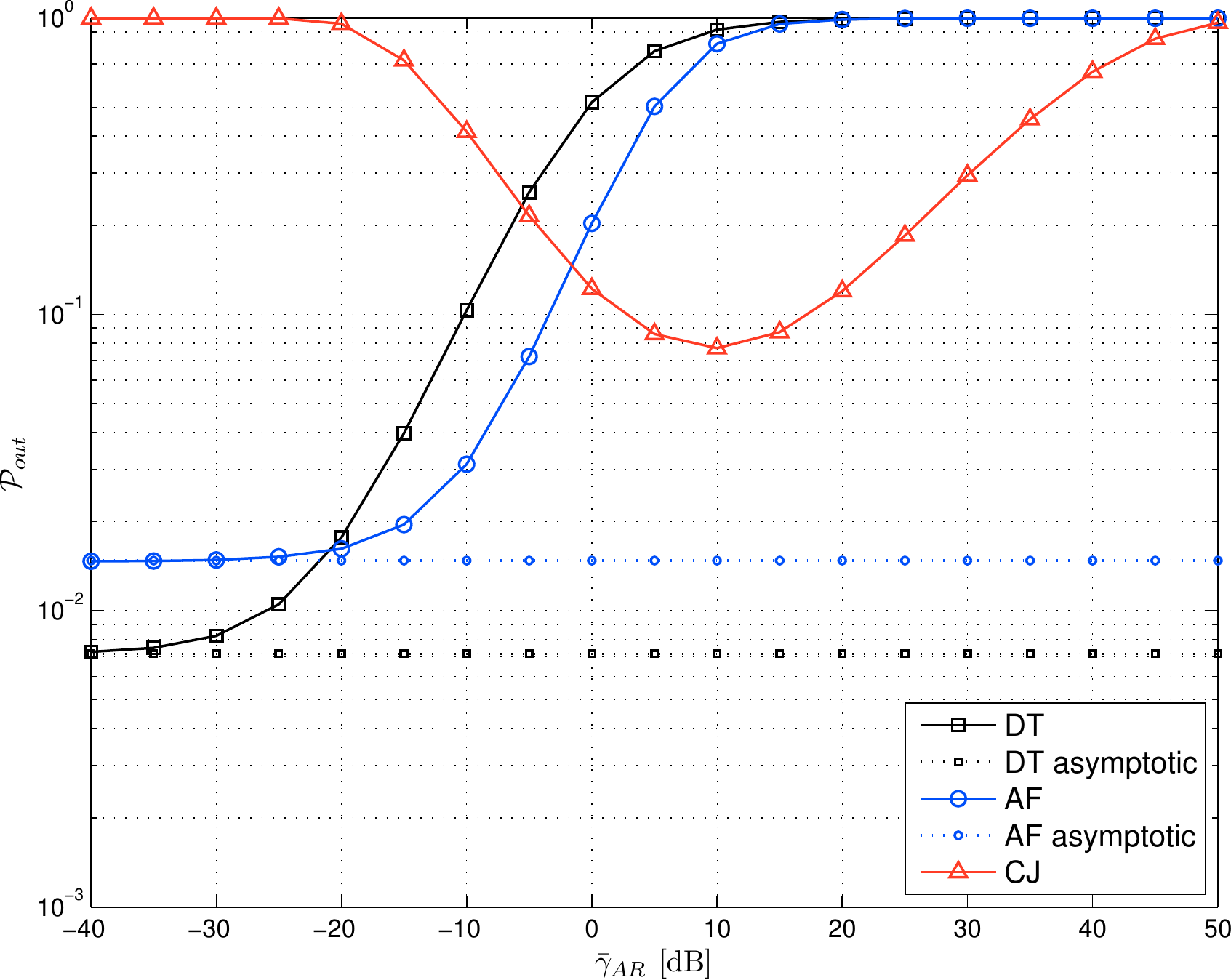}
\caption{\label{fig:osrar}Outage probability versus $\bar{\gamma}_{A,R}$,
single-antenna relay, $\bar{\gamma}_{A,B}=0$dB, $\bar{\gamma}_{R,B}=5$dB,
$\rho=20$dB, asymptotic results computed with Eqs.~\eqref{eq:34}
for DT, and \eqref{eq:35} for AF.}
\end{figure}

\begin{figure}[h]
\centering \includegraphics[width=0.55\textwidth]{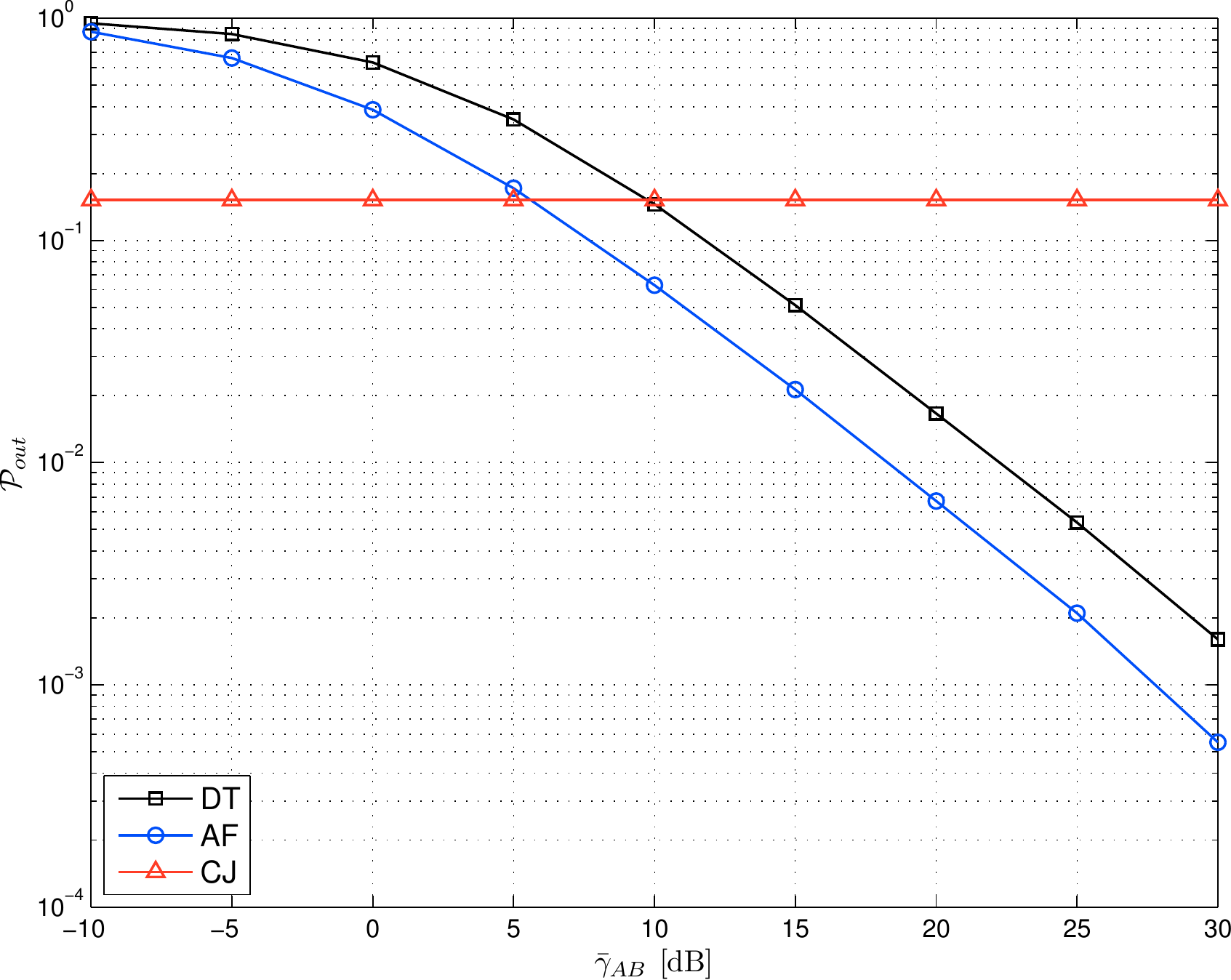}
\caption{\label{fig:osrab}Outage probability versus $\bar{\gamma}_{A,B}$,
single-antenna relay, $\bar{\gamma}_{A,R}=2$dB, $\bar{\gamma}_{R,B}=10$dB,
$\rho=10$dB.}
\end{figure}

In Fig.~\ref{fig:osrabrrb}, the outage performance is shown when
$\bar{\gamma}_{A,B}$ and $\bar{\gamma}_{R,B}$ both increase simultaneously,
where $\bar{\gamma}_{A,R}=2$dB and $\rho=10$dB. Note that since
the performance of CJ does not depend on the direct link, the asymptotic
SOP of CJ can still be characterized via Eq.~\eqref{eq:30}, which
indicates that the SOP of CJ will converge to a constant. On the other
hand, we see that AF outperforms the other schemes since its outage
probability decays to zero faster. This is due to the fact that when
only $\bar{\gamma}_{A,B}$ increases, the outage probability of DT
and AF decays with the same slope (see Fig.~\ref{fig:osrab}), and
when $\bar{\gamma}_{R,B}$ also increases at the same time, AF will
enjoy a better second hop channel (and thus the outage probability
decays faster), which does not benefit DT.

\begin{figure}[h]
\centering \includegraphics[width=0.55\textwidth]{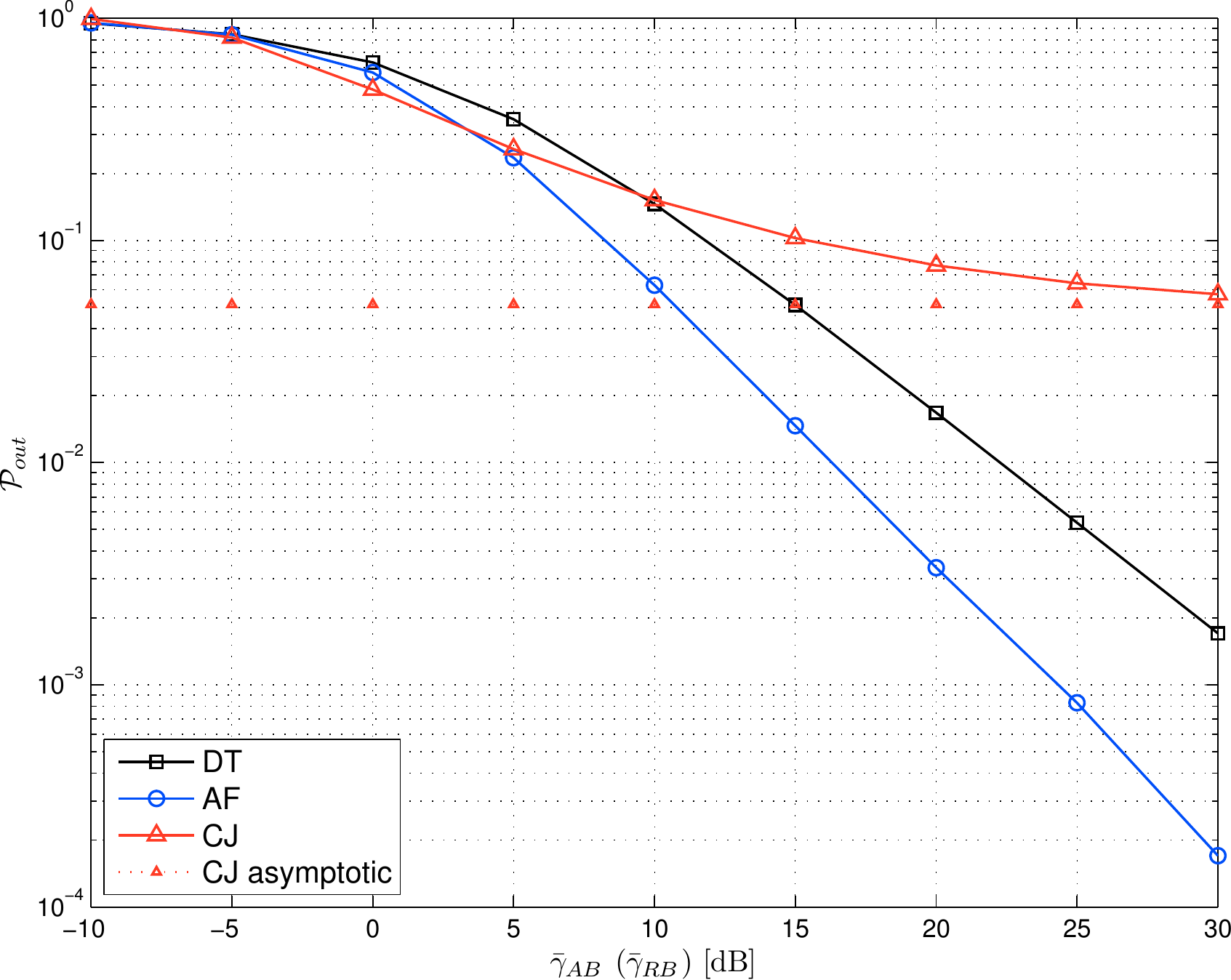}
\caption{\label{fig:osrabrrb}Outage probability versus $\bar{\gamma}_{A,B}$,
single-antenna relay, $\bar{\gamma}_{R,B}=\bar{\gamma}_{A,B}$, $\bar{\gamma}_{A,R}=2$dB,
$\rho=10$dB, asymptotic results computed with Eq.~\eqref{eq:30}.}
\end{figure}

\subsection{Multi-Antenna Case}

Fig.~\ref{fig:multiant_k} compares the multi-antenna SOP as a
function of the number of relay antennas $K$ for average channel
gains $\bar{\gamma}_{A,B}=5$dB, $\bar{\gamma}_{A,R}=0$dB and
$\bar{\gamma}_{R,B}=10$dB with $\rho=30$dB.  As seen from the DT
and AF curves, the exact analytical SOP for the multi-antenna
scenario derived in \eqref{eq:39} and \eqref{eq:41} agrees very
well with the simulations. When all available antennas are used
at the untrusted relay, the figure shows that the SOPs of all
schemes converge to unity as $K$ grows, as predicted in Section
\ref{sec:multi-antenna-relay}. Moreover, we see that when $K$
changes from 1 to 2, the SOP of CJ increases rapidly, because
multi-antenna receive beamforming at the relay can suppress the
intentional interference from Bob, rendering CJ
ineffective. Also, the performance achieved by an optimal power 
allocation is shown in the figure, and we see a slight reduction
in the outage probability for both AF and CJ. Consistent
with the result in Fig.~\ref{fig:osp}, the performance gain of
power allocation for CJ is not obvious when $K=1$.

\begin{figure}[h]
\centering \includegraphics[width=0.55\textwidth]{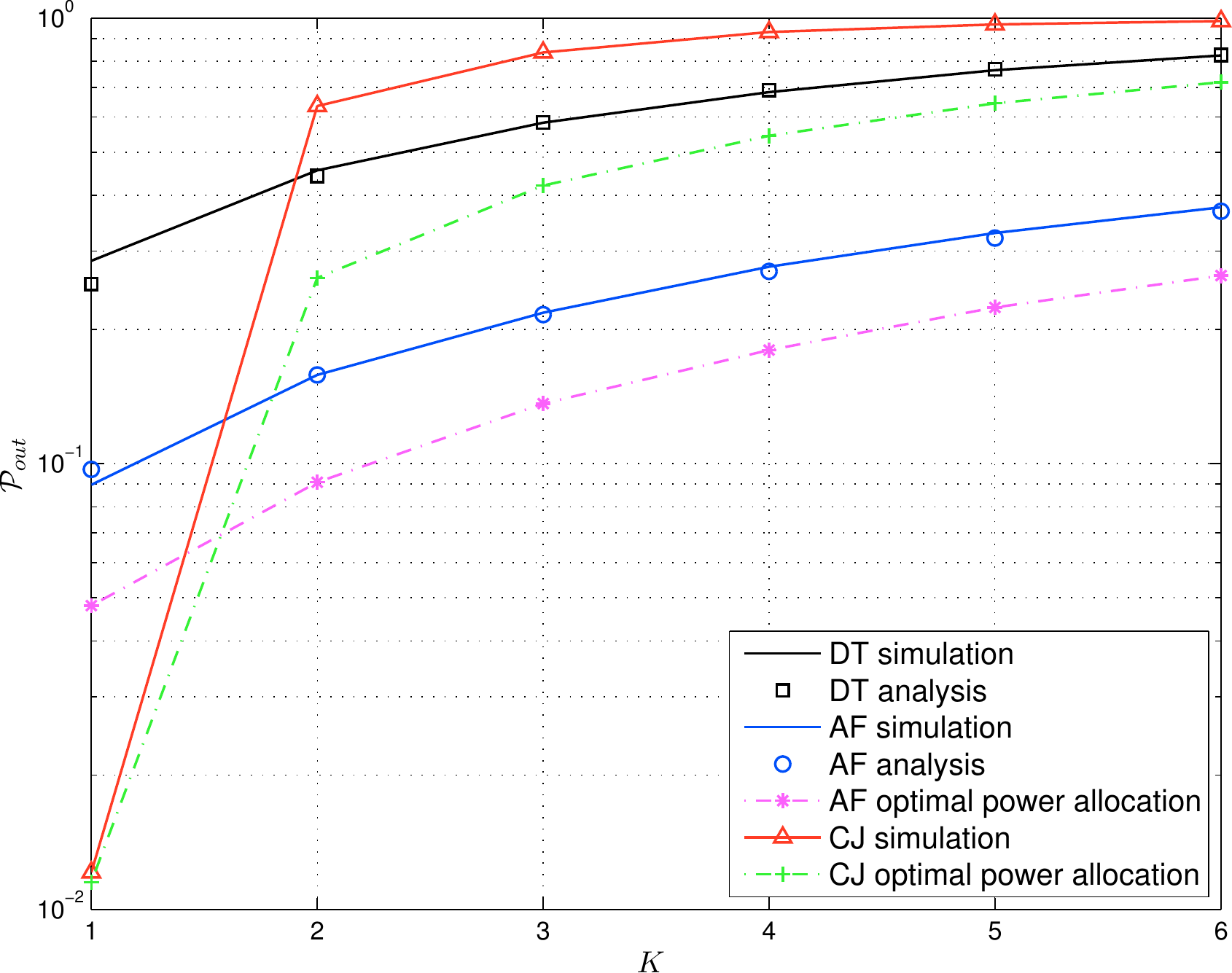}
\caption{\label{fig:multiant_k} Outage probability versus number of relay
antennas, multi-antenna relay, $\bar{\gamma}_{A,B}=5$dB, $\bar{\gamma}_{A,R}=0$dB,
$\bar{\gamma}_{R,B}=10$dB, $\rho=30$dB, analytical results computed
with Eqs.~\eqref{eq:39} for DT, \eqref{eq:41} for AF.}
\end{figure}

The secrecy performance for various relaying schemes with antenna
selection is shown in Figs.~\ref{fig:antsel_p} and \ref{fig:antsel_k_compare}.
In Fig.~\ref{fig:antsel_p}, the SOP is evaluated as a function of
$\rho$ for $\bar{\gamma}_{A,B}=5$dB, $\bar{\gamma}_{A,R}=0$dB and
$\bar{\gamma}_{R,B}=5$dB with six antennas employed at the relay.
We see that the analytical results derived in \eqref{eq:55}, \eqref{eq:59}
and \eqref{eq:68} respectively match the simulations for DT, AF and CJ without
the second-hop CSI. The schemes with antenna selection
show properties similar to those for a single-antenna relay in Fig.~\ref{fig:osp}
as $\rho$ increases; \textit{i.e.} the SOP of DT and AF converges
to constants and that of CJ decays to zero. As expected, CJ with second-hop
CSI decays to zero faster than CJ without CSI, since in the former approach the best transmit antenna at the relay is chosen, and such diversity
gain is more obvious with larger $\rho$, as seen in the numerator
of \eqref{eq:64}.

The SOP of relaying schemes with and without antenna selection for
increasing $K$ is depicted in Fig.~\ref{fig:antsel_k_compare},
and again we see that the SOP of DT and AF converges to unity. For
both the AF and CJ schemes, the antenna selection schemes demonstrate
lower SOP. This is due to the fact that the relay loses the array gain under antenna selection, and this
gain is more beneficial to the relay than to Bob. We also see the
significantly improved secrecy that results for CJ when the relay
can only perform antenna selection instead of beamforming. Interestingly,
the SOP of CJ with second-hop CSI decreases first and then gradually
increases as the number of relay antennas grows. This is because the
second-hop diversity gain at first outweighs the first-hop diversity
for the relay, and then there is a diminishing marginal return for
the second-hop with larger $K$ ( $\frac{|h_{n^{*},B}|^{2}}{|h_{n^{*},B}|^{2}+\bar{\gamma}_{A,R}+\bar{\gamma}_{R,B}+\frac{1}{\rho}}\rightarrow1$
as $K\rightarrow\infty$ in \eqref{eq:64}) and the secrecy performance
gradually degrades. However, when the relay does not have second-hop
CSI, the SOP monotonically decreases with $K$ and converges to $1-e^{-\frac{t}{\bar{\gamma}_{R,B}}}$,
which validates Corollary \ref{sec:relay-with-second}. Therefore,
for large $K$, in order to maintain confidentiality, CJ should be
used with the second-hop CSI concealed from the untrusted relay.

\begin{figure}[h]
\centering \includegraphics[width=0.55\textwidth]{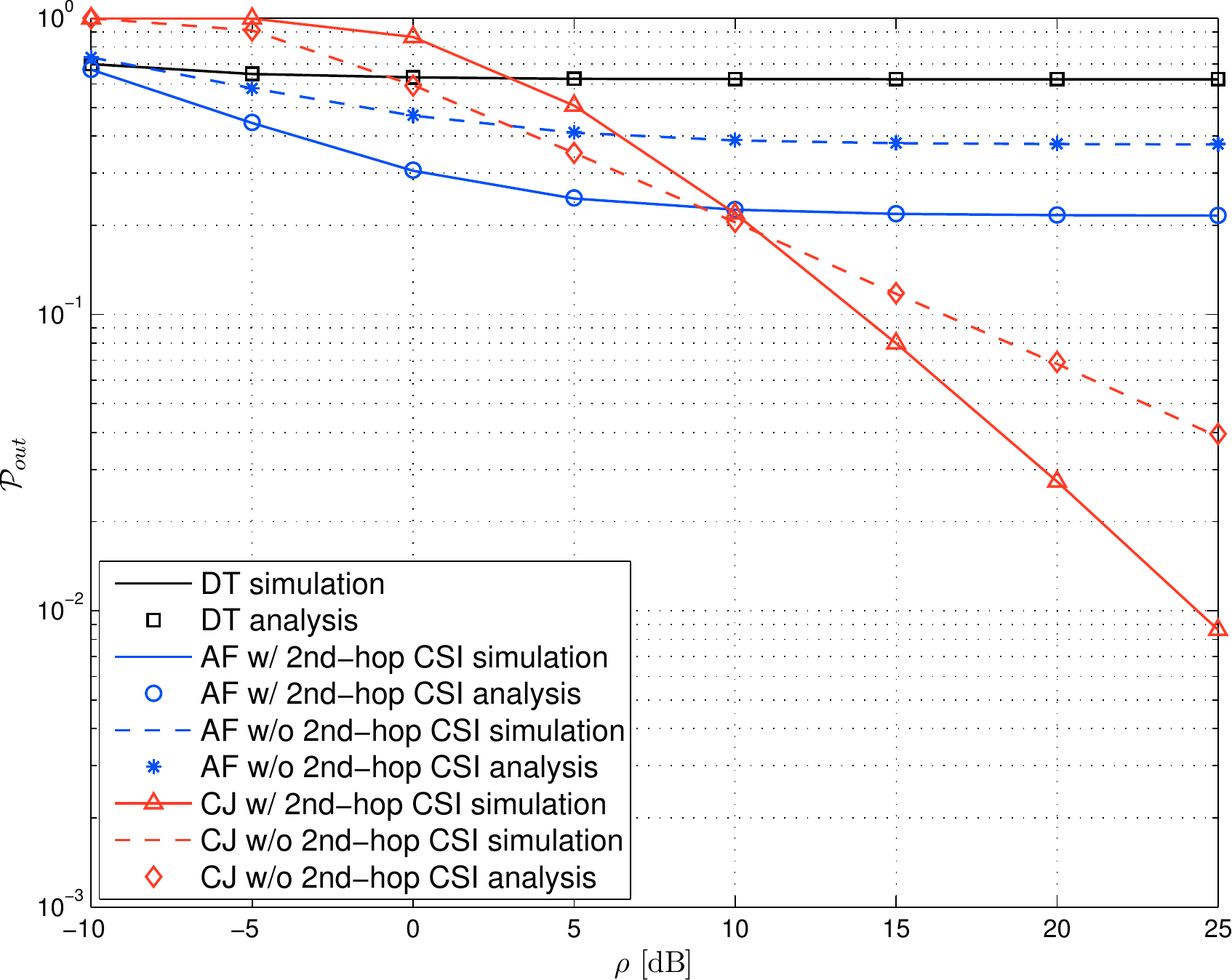}
\caption{\label{fig:antsel_p} Outage probability versus $\rho$, multi-antenna
relay ($K=6$) with antenna selection , $\bar{\gamma}_{A,B}=5$dB,
$\bar{\gamma}_{A,R}=0$dB, $\bar{\gamma}_{R,B}=5$dB, analytical results
computed with Eqs.~\eqref{eq:55} for DT, \eqref{eq:59} and \eqref{eq:63}
for AF, and \eqref{eq:68} for CJ.}
\end{figure}

\begin{figure}[h]
\centering \includegraphics[width=0.55\textwidth]{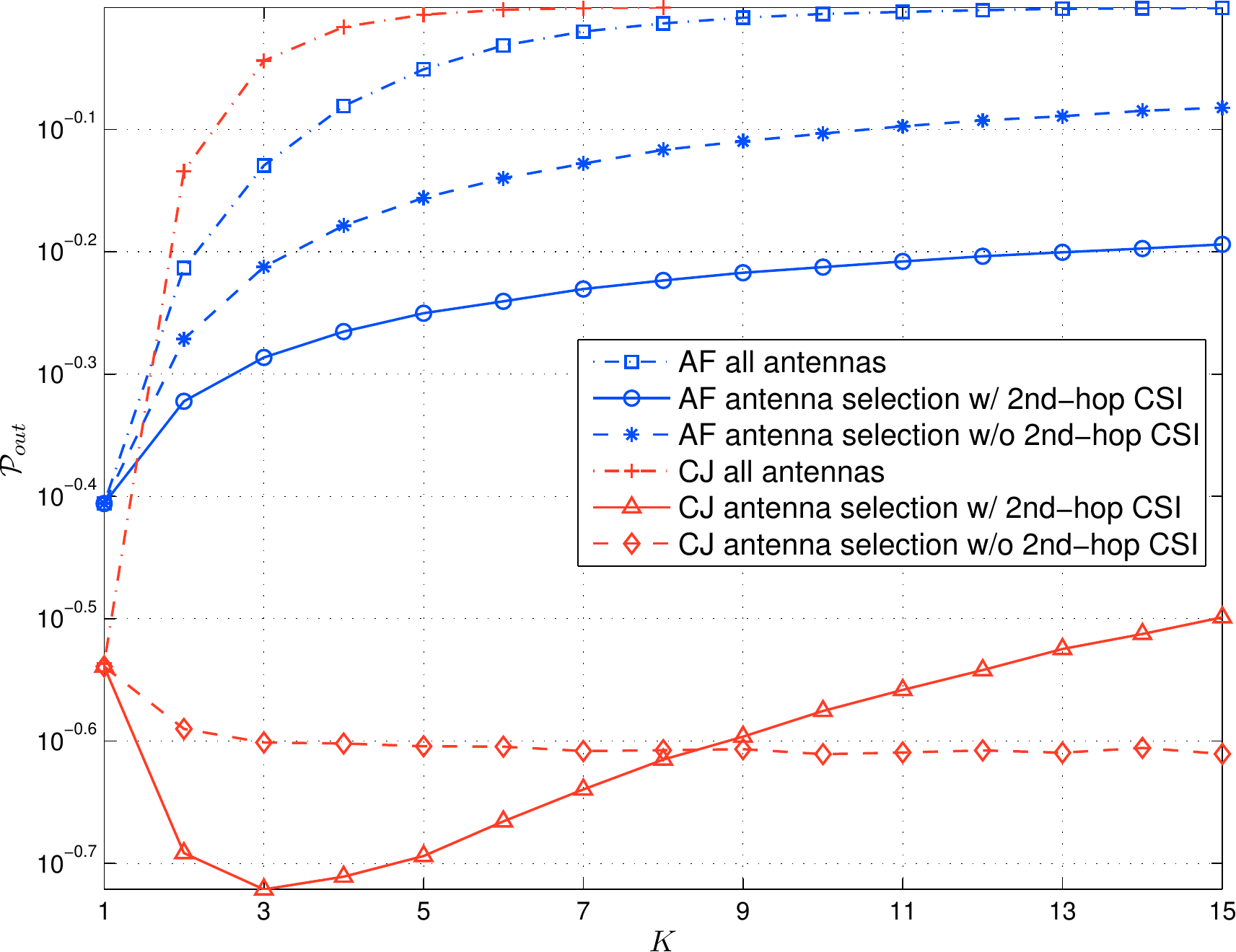}
\caption{\label{fig:antsel_k_compare} Outage probability versus number of
relay antennas, $\bar{\gamma}_{A,B}=\bar{\gamma}_{A,R}=0$dB, $\bar{\gamma}_{R,B}=2$dB,
$\rho=12$dB.}
\end{figure}

\section{Conclusions}

\label{sec:con}

This paper has analyzed a three-node network where the source can
potentially utilize an untrusted multi-antenna relay to supplement
the direct link to its destination. The untrusted relay is in effect
an eavesdropper, although also assisting the source with cooperative
transmission. We derived the exact secrecy outage probability of three
different transmission policies: direct transmission without using
the relay, conventional non-regenerative relaying, and cooperative
jamming by the destination. The SOP computation allows performance
transitions between the three algorithms to be determined for different
scenarios. An asymptotic analysis of the outage probabilities is also
conducted to elicit the optimal policies for different operating regimes.
When the relay has a large number of antennas, we showed that the
SOP for all three approaches converges to unity. However, when antenna
selection is used at the relay, the secrecy performance for all schemes
is improved, and the CJ scheme in particular can obtain a significant
diversity gain with a moderate growth in the number of antennas. Moreover,
if the destination conceals its CSI from the relay, secrecy will not
be compromised as the antenna number grows. Our theoretical predictions
were validated via various numerical examples.

\appendices{ }

\section{Probability of positive secrecy rate for AF}

\label{sec:prob-posit-secr}

Assuming all channels $\gamma_{ij}$ are independent, define two random
variables $U$ and $V$ as
\begin{align*}
U & =\frac{|h_{A,B}|^{2}}{|h_{A,R}|^{2}},\quad V=\frac{|h_{R,B}|^{2}}{|h_{R,B}|^{2}+\bar{\gamma}_{A,R}+\frac{1}{\rho}}.
\end{align*}
The c.d.f. of $U$ is given by
\begin{align}
F_{U}(u) & =\mathcal{P}\left\{ \frac{|h_{A,B}|^{2}}{|h_{A,R}|^{2}}\le u\right\} \notag\label{eq:70}\\
 & =\int_{0}^{\infty}\int_{0}^{yu}p_{|h_{A,B}|^{2}}(x)p_{|h_{A,R}|^{2}}(y)~dxdy\notag\\
 & =\frac{\bar{\gamma}_{A,R}u}{\bar{\gamma}_{A,B}+\bar{\gamma}_{A,R}u}.
\end{align}
Differentiating $F_{U}(u)$ with respect to $u$, we obtain the p.d.f.
of $U$ as $p_{U}(u)=\frac{\bar{\gamma}_{A,B}\bar{\gamma}_{A,R}}{(\bar{\gamma}_{A,B}+\bar{\gamma}_{A,R}u)^{2}}$.
For $V$, the c.d.f. is given by
\begin{equation}
F_{V}(v)=\mathcal{P}\left\{ \frac{|h_{R,B}|^{2}}{|h_{R,B}|^{2}+\bar{\gamma}_{A,R}+\frac{1}{\rho}}\le v\right\} =\mathcal{P}\left\{ |h_{R,B}|^{2}\le g(v)\right\} =1-e^{-\frac{g(v)}{\bar{\gamma}_{R,B}}}\label{eq:71}
\end{equation}
where $g(v)=\frac{(\bar{\gamma}_{A,R}+\frac{1}{\rho})v}{1-v}$. Differentiating
$F_{V}(v)$ with respect to $v$, we have
\begin{equation}
p_{V}(v)=\frac{\bar{\gamma}_{A,R}+\frac{1}{\rho}}{(1-v)^{2}\bar{\gamma}_{R,B}}e^{-\frac{(\bar{\gamma}_{A,R}+\frac{1}{\rho})v}{\bar{\gamma}_{R,B}(1-v)}}.\label{eq:72}
\end{equation}

Next, we can calculate $\mathcal{P}_{pos}^{AF}$ as
\begin{align}
\mathcal{P}_{pos}^{AF} & =\mathcal{P}\left\{ \frac{|h_{A,B}|^{2}}{|h_{A,R}|^{2}}+\frac{|h_{R,B}|^{2}}{|h_{R,B}|^{2}+\bar{\gamma}_{A,R}+\frac{1}{\rho}}>1\right\} \notag\label{eq:73}\\*
 & =\mathcal{P}\left\{ U+V>1\right\} .
\end{align}
Since $U$ and $V$ are independent,
\begin{align}
 & \mathcal{P}_{pos}^{AF}\notag=\int_{0}^{1}\int_{1-v}^{\infty}p_{U,V}(u,v)~dudv=\int_{0}^{1}\int_{1-v}^{\infty}p_{U}(u)p_{V}(v)~dudv\notag\\
 & =\frac{\bar{\gamma}_{A,B}\bar{\gamma}_{A,R}(\bar{\gamma}_{A,R}+\frac{1}{\rho})}{\bar{\gamma}_{R,B}}\int_{0}^{1}\int_{1-v}^{\infty}\frac{1}{(\bar{\gamma}_{A,B}+\bar{\gamma}_{A,R}u)^{2}}\frac{1}{(1-v)^{2}}e^{-\frac{(\bar{\gamma}_{A,R}+\frac{1}{\rho})v}{\bar{\gamma}_{R,B}(1-v)}}~dudv\notag\\
 & =\frac{\bar{\gamma}_{A,B}\bar{\gamma}_{A,R}(\bar{\gamma}_{A,R}+\frac{1}{\rho})}{\bar{\gamma}_{R,B}}\int_{0}^{1}\frac{1}{(\bar{\gamma}_{A,B}\bar{\gamma}_{A,R}+\bar{\gamma}_{A,R}^{2}-\bar{\gamma}_{A,R}^{2}v)}\frac{1}{(1-v)^{2}}e^{-\frac{(\bar{\gamma}_{A,R}+\frac{1}{\rho})v}{\bar{\gamma}_{R,B}(1-v)}}~dv\notag\\
 & =\mu_{1}\int_{0}^{\infty}\frac{x+1}{x+\beta_{1}}e^{-\mu_{1}x}~dx\label{eq:74}\\
 & =\mu_{1}\left[\int_{0}^{\infty}\frac{x}{x+\beta_{1}}e^{-\mu_{1}x}~dx+\int_{0}^{\infty}\frac{1}{x+\beta_{1}}e^{-\mu_{1}x}~dx\right]\notag\\
 & =\mu_{1}(\beta_{1}-1)e^{\mu_{1}\beta_{1}}\mathrm{Ei}(-\mu_{1}\beta_{1})+1,\label{eq:75}
\end{align}
where in \eqref{eq:74}, we use the transformation $x=\frac{v}{1-v}$,
$\mu_{1}=\frac{\bar{\gamma}_{A,R}+1/\rho}{\bar{\gamma}_{R,B}}$ and
$\beta_{1}=1+\frac{\bar{\gamma}_{A,R}}{\bar{\gamma}_{A,B}}$. Eq.~\eqref{eq:75}
is obtained using the identity in \cite[eq.~3.353.5]{Gradshteyn_Tables00}.

\section{Proof of Proposition \ref{sec:af}}

\label{sec:secr-outage-prob} Let $X=|h_{A,B}|^{2}$ and $Y=|h_{A,R}|^{2}$
be exponentially distributed random variables and define
\[
Z=\frac{1+\rho|h_{A,B}|^{2}+\rho\frac{|h_{R,B}|^{2}|h_{A,R}|^{2}}{|h_{R,B}|^{2}+\bar{\gamma}_{A,R}+\frac{1}{\rho}}}{1+\rho|h_{A,R}|^{2}},\quad V=\frac{|h_{R,B}|^{2}}{|h_{R,B}|^{2}+\bar{\gamma}_{A,R}+\frac{1}{\rho}}
\]
where the p.d.f. of $V$ is given by \eqref{eq:72}. We thus have
\begin{align}
{P}_{out}^{AF} & =\mathcal{P}\left(\frac{1+\rho X+\rho YV}{1+\rho Y}<2^{2R}\right)\notag\\
 & =\mathbb{E}_{V}\{\mathbb{E}_{Y}\{F_{Z|Y,V}(2^{2R})\}\}\notag\\
 & =\int_{0}^{1}\int_{0}^{\infty}\left\{ 1-e^{-\frac{1}{\bar{\gamma}_{A,B}}\left(\frac{2^{2R}-1}{\rho}+(2^{2R}-v)y\right)}\right\} p_{Y}(y)p_{V}(v)~dydv\notag\\
 & =1-\frac{\bar{\gamma}_{A,R}+\frac{1}{\rho}}{\bar{\gamma}_{A,R}\bar{\gamma}_{R,B}}e^{-\frac{2^{2R}-1}{\rho\bar{\gamma}_{A,B}}}\int_{0}^{1}\int_{0}^{\infty}\frac{1}{(1-v)^{2}}~e^{-\frac{(2^{2R}-v)y}{\bar{\gamma}_{A,B}}-\frac{y}{\bar{\gamma}_{A,R}}-\frac{(\bar{\gamma}_{A,R}+\frac{1}{\rho})v}{\bar{\gamma}_{R,B}(1-v)}}~dydv\notag\\
 & =1-\frac{\bar{\gamma}_{A,B}(\bar{\gamma}_{A,R}+\frac{1}{\rho})}{\bar{\gamma}_{R,B}}e^{-\frac{2^{2R}-1}{\rho\bar{\gamma}_{A,B}}}\int_{0}^{1}\frac{1}{(2^{2R}\bar{\gamma}_{A,R}+\bar{\gamma}_{A,B}-\bar{\gamma}_{A,R}v)(1-v)^{2}}e^{-\frac{(\bar{\gamma}_{A,R}+\frac{1}{\rho})v}{\bar{\gamma}_{R,B}(1-v)}}~dv\notag\\
 & =1-\frac{\mu_{1}\bar{\gamma}_{A,B}}{(2^{2R}-1)\bar{\gamma}_{A,R}+\bar{\gamma}_{A,B}}e^{-\frac{2^{2R}-1}{\rho\bar{\gamma}_{A,B}}}\int_{0}^{\infty}\frac{x+1}{x+\beta_{2}}e^{-\mu_{1}x}~dx\label{eq:76}\\
 & =1-\frac{\bar{\gamma}_{A,B}}{(2^{2R}-1)\bar{\gamma}_{A,R}+\bar{\gamma}_{A,B}}e^{-\frac{2^{2R}-1}{\rho\bar{\gamma}_{A,B}}}\left[\mu_{1}(\beta_{2}-1)e^{\mu_{1}\beta_{2}}\mathrm{Ei}(-\mu_{1}\beta_{2})+1\right]\label{eq:77}
\end{align}
where in \eqref{eq:76}, we use the transformation $x=\frac{v}{1-v}$,
$\mu_{1}=\frac{\bar{\gamma}_{A,R}+1/\rho}{\bar{\gamma}_{R,B}}$, $\beta_{2}=\frac{2^{2R}\bar{\gamma}_{A,R}+\bar{\gamma}_{A,B}}{(2^{2R}-1)\bar{\gamma}_{A,R}+\bar{\gamma}_{A,B}}$,
and the result in \eqref{eq:74}-\eqref{eq:75} is applied to obtain
\eqref{eq:77}.

\section{Proof of Proposition \ref{sec:coop-jamm-cj}}

\label{sec:deriv-eqref}

The outage probability for a given secrecy rate $R$ is given by
\begin{align}
\mathcal{P}_{out}^{CJ} & =\mathcal{P}\left(\frac{1+\rho\frac{|h_{R,B}|^{2}|h_{A,R}|^{2}}{|h_{R,B}|^{2}+\bar{\gamma}_{A,R}+\bar{\gamma}_{R,B}+\frac{1}{\rho}}}{1+\frac{|h_{A,R}|^{2}}{|h_{R,B}|^{2}+\frac{1}{\rho}}}<2^{2R}\right)\notag\\
 & =\mathcal{P}\left(\phi(|h_{R,B}|^{2})|h_{A,R}|^{2}<2^{2R}-1\right),\label{eq:78}
\end{align}
where
\begin{equation}
\phi(z)=\frac{\rho z}{z+\bar{\gamma}_{A,R}+\bar{\gamma}_{R,B}+\frac{1}{\rho}}-\frac{2^{2R}}{z+\frac{1}{\rho}}.\label{eq:79}
\end{equation}
Then, we have
\begin{align}
\mathcal{P}_{out}^{CJ} & =\int_{t}^{\infty}F_{|h_{A,R}|^{2}}\left(\frac{2^{2R}-1}{\phi(z)}\right)p_{|h_{R,B}|^{2}}(z)~dz+\int_{0}^{t}\left[1-F_{|h_{A,R}|^{2}}\left(\frac{2^{2R}-1}{\phi(z)}\right)\right]p_{|h_{R,B}|^{2}}(z)~dz\notag\label{eq:80}\\
 & =\int_{t}^{\infty}\left[1-e^{-\frac{2^{2R}-1}{\bar{\gamma}_{A,R}\phi(z)}}\right]p_{|h_{R,B}|^{2}}(z)~dz+\int_{0}^{t}p_{|h_{R,B}|^{2}}(z)~dz\notag\\
 & =1-\frac{1}{\bar{\gamma}_{R,B}}\int_{t}^{\infty}e^{-\frac{2^{2R}-1}{\bar{\gamma}_{A,R}\phi(z)}-\frac{z}{\bar{\gamma}_{R,B}}}~dz
\end{align}
where
\begin{equation}
\phi(z)\left\{ \begin{array}{ll}
\ge0, & z\ge t\\
<0, & 0\le z<t
\end{array}\right.\label{eq:81}
\end{equation}
and
\begin{equation}
t=\frac{(2^{2R}-1)+\sqrt{(2^{2R}-1)^{2}+\rho2^{2R+1}(\bar{\gamma}_{A,R}+\bar{\gamma}_{R,B}+1/\rho)}}{2\rho}.\label{eq:82}
\end{equation}

\section{Proof of Lemma \ref{sec:case-bargamma_r-b}}

\label{sec:proof-lemma-refs}

Define $X=\rho\frac{|h_{R,B}|^{2}|h_{A,R}|^{2}}{|h_{R,B}|^{2}+\bar{\gamma}_{A,R}+
\bar{\gamma}_{R,B}+\frac{1}{\rho}}$.  The outage probability of CJ with $\bar{\gamma}_{R,B}\rightarrow\infty$
is written as
\begin{align*}
 & \lim_{\bar{\gamma}_{R,B}\rightarrow\infty}\mathcal{P}_{out}^{CJ}=\lim_{\bar{\gamma}_{R,B}\rightarrow\infty}\mathcal{P}\left(\frac{1+X}{1+\frac{|h_{A,R}|^{2}}{|h_{R,B}|^{2}+\frac{1}{\rho}}}<2^{2R}\right)\\
 & =\lim_{\bar{\gamma}_{R,B}\rightarrow\infty}\left\{ \mathcal{P}\left(\left.\frac{1+X}{1+\frac{|h_{A,R}|^{2}}{|h_{R,B}|^{2}+\frac{1}{\rho}}}<2^{2R}\right||h_{R,B}|^{2}\ge\sqrt{\bar{\gamma}_{R,B}}\right)\mathcal{P}(|h_{R,B}|^{2}\ge\sqrt{\bar{\gamma}_{R,B}})\right.\\
 & \left.+\mathcal{P}\left(\left.\frac{1+X}{1+\frac{|h_{A,R}|^{2}}{|h_{R,B}|^{2}+\frac{1}{\rho}}}<2^{2R}\right||h_{R,B}|^{2}<\sqrt{\bar{\gamma}_{R,B}}\right)\mathcal{P}(|h_{R,B}|^{2}<\sqrt{\bar{\gamma}_{R,B}})\right\} .
\end{align*}
Since $|h_{R,B}|^{2}$ follows the exponential distribution, \emph{i.e.}
$|h_{R,B}|^{2}\sim\exp\left(\frac{1}{\bar{\gamma}_{R,B}}\right)$,
we have $\lim_{\bar{\gamma}_{R,B}\rightarrow\infty}\mathcal{P}(|h_{R,B}|^{2}\ge\sqrt{\bar{\gamma}_{R,B}})=\lim_{\bar{\gamma}_{R,B}\rightarrow\infty}e^{\frac{-1}{\sqrt{\bar{\gamma}_{R,B}}}}=1$
and $\lim_{\bar{\gamma}_{R,B}\rightarrow\infty}\mathcal{P}(|h_{R,B}|^{2}<\sqrt{\bar{\gamma}_{R,B}})=1-\lim_{\bar{\gamma}_{R,B}\rightarrow\infty}e^{\frac{-1}{\sqrt{\bar{\gamma}_{R,B}}}}=0$.
So the SOP of CJ is given by
\[
\lim_{\bar{\gamma}_{R,B}\rightarrow\infty}\mathcal{P}_{out}^{CJ}=\lim_{\bar{\gamma}_{R,B}\rightarrow\infty}\mathcal{P}\left(\left.\frac{1+X}{1+\frac{|h_{A,R}|^{2}}{|h_{R,B}|^{2}+\frac{1}{\rho}}}<2^{2R}\right||h_{R,B}|^{2}\ge\sqrt{\bar{\gamma}_{R,B}}\right).
\]
Under the condition that $|h_{R,B}|^{2}\ge\sqrt{\bar{\gamma}_{R,B}}$,
we have
\[
\frac{1+X}{1+\frac{|h_{A,R}|^{2}}{\sqrt{\bar{\gamma}_{R,B}}+\frac{1}{\rho}}}\le\frac{1+X}{1+\frac{|h_{A,R}|^{2}}{|h_{R,B}|^{2}+\frac{1}{\rho}}}\le1+X
\]
and consequently,
\begin{equation}
\mathcal{P}\left\{ \frac{1+X}{1+\frac{|h_{A,R}|^{2}}{\sqrt{\bar{\gamma}_{R,B}}+\frac{1}{\rho}}}<2^{2R}\right\} \ge\mathcal{P}\left\{ \left.\frac{1+X}{1+\frac{|h_{A,R}|^{2}}{|h_{R,B}|^{2}+\frac{1}{\rho}}}<2^{2R}\right||h_{R,B}|^{2}\ge\sqrt{\bar{\gamma}_{R,B}}\right\} \ge\mathcal{P}\left\{ 1+X<2^{2R}\right\} .\label{eq:43}
\end{equation}
Because $\mathcal{P}\left\{ \frac{1+X}{1+\frac{|h_{A,R}|^{2}}{\sqrt{\bar{\gamma}_{R,B}}+\frac{1}{\rho}}}<2^{2R}\right\} $
is continuous with respect to $\bar{\gamma}_{R,B}$, which can be
verified by showing that the function inside the probability integral
is differentiable (results in Appendix \ref{sec:deriv-eqref} can
be reused, but we skip the details here due to space constraints),
we have
\begin{equation}
\lim_{\bar{\gamma}_{R,B}\rightarrow\infty}\mathcal{P}\left\{ \frac{1+X}{1+\frac{|h_{A,R}|^{2}}{\sqrt{\bar{\gamma}_{R,B}}+\frac{1}{\rho}}}<2^{2R}\right\} =\lim_{\bar{\gamma}_{R,B}\rightarrow\infty}\mathcal{P}\left\{ 1+X<2^{2R}\right\} .\label{eq:38}
\end{equation}
Combining \eqref{eq:43} and \eqref{eq:38}, the conclusion in Lemma
\ref{sec:case-bargamma_r-b} can be inferred and the proof is completed.

%========================================================================
%========================================================================
{\linespread{1.30}
 \bibliographystyle{IEEEtran}
\bibliography{IEEEabrv,mybibfile}
 }
%\newpage

\end{document}